\documentclass[12pt]{amsart}
\usepackage{a4wide}
\usepackage[latin9]{inputenc}
\usepackage{amsmath}
\usepackage{amsfonts}
\usepackage{amssymb}
\usepackage{amsthm}
\usepackage{subfigure}
\newtheorem{theo}{Theorem}[section]
\newtheorem{prop}[theo]{Proposition}

\newtheorem{lemm}[theo]{Lemma}

\theoremstyle{definition}

\theoremstyle{remark}
\newtheorem{rema}[theo]{Remark}
\newcommand{\Op}{\operatorname{Op}}

\newcommand{\nwc}{\newcommand}
\nwc{\eps}{\epsilon}
\nwc{\ep}{\epsilon}
\nwc{\vareps}{\varepsilon}
\nwc{\Oph}{\operatorname{Op}_\hbar}
\nwc{\la}{\langle}
\nwc{\ra}{\rangle}

\nwc{\mf}{\mathbf} 
\nwc{\blds}{\boldsymbol} 
\nwc{\ml}{\mathcal} 

\nwc{\defeq}{\stackrel{\rm{def}}{=}}

\nwc{\cE}{\ml{E}}
\nwc{\cN}{\ml{N}}
\nwc{\cO}{\ml{O}}
\nwc{\cP}{\ml{P}}
\nwc{\cU}{\ml{U}}
\nwc{\cV}{\ml{V}}
\nwc{\cW}{\ml{W}}
\nwc{\tU}{\widetilde{U}}
\nwc{\IN}{\mathbb{N}}
\nwc{\IR}{\mathbb{R}}
\nwc{\IS}{\mathbb{S}}
\nwc{\IZ}{\mathbb{Z}}
\nwc{\IC}{\mathbb{C}}
\nwc{\IT}{\mathbb{T}}
\nwc{\be}{\mathbf{e}}
\nwc{\tP}{\widetilde{P}}
\nwc{\tPi}{\widetilde{\Pi}}
\nwc{\tV}{\widetilde{V}}
\nwc{\supp}{\operatorname{supp}}
\nwc{\rest}{\restriction}

\newcommand{\Vol}{\operatorname{Vol}}

\date{\today}

\begin{document}

\title[The Quantum Loschmidt echo on flat tori]{The Quantum Loschmidt echo on flat tori}

\author[Gabriel Rivi\`ere]{Gabriel Rivi\`ere}
\author[Henrik Uebersch\"ar]{Henrik Uebersch\"ar}

\address{Laboratoire Paul Painlev\'e (U.M.R. CNRS 8524), U.F.R. de Math\'ematiques, Universit\'e Lille 1, 59655 Villeneuve d'Ascq Cedex, France}
\email{gabriel.riviere@math.univ-lille1.fr}
\address{Institut de Math\'ematiques de Jussieu (U.M.R. CNRS 7586), U.F.R. de Math\'ematiques, Universit\'e Paris 6 -- Pierre et Marie Curie, Sorbonne Universit\'es, 75252 Paris Cedex 05, France}
\email{henrik.ueberschar@imj-prg.fr}

\begin{abstract} The Quantum Loschmidt Echo is a measurement of the sensitivity of a quantum system to perturbations of the Hamiltonian. 
In the case of the standard $2$-torus, we derive some explicit formulae for this quantity in the transition regime where it is 
expected to decay in the semiclassical limit. The expression involves both a two-microlocal defect measure of the initial data 
and the form of the perturbation. As an application, we exhibit a non-concentration criterium on the sequence of initial data 
under which one does not observe a macroscopic decay of the Quantum Loschmidt Echo. We also apply our results to several 
examples of physically relevant initial data such as coherent states and plane waves.

\end{abstract}

\maketitle

\section{Introduction}

Let $\IT^2=\IR^2/\IZ^2$ be the standard\footnote{We consider the standard torus for simplicity and our analysis would extend to 
$\IR^2/\Gamma$ with $\Gamma= a\IZ\oplus b\IZ$ with $a,b>0$.} torus. Motivated by the fact that the quantum evolution is unitary and thus cannot be sensitive 
to perturbations of initial conditions, Peres suggested to study the sensitivity of the Schr\"odinger 
equation to perturbations of the Hamiltonian~\cite{Pe84}. More precisely, he proposed to compare 
the dynamics induced by the following two semiclassical Schr\"odinger equations:
\begin{equation}\label{e:schrodinger-unperturbed}i\hbar\partial_tu_{\hbar}=-\frac{\hbar^2\Delta u_{\hbar}}{2},\quad u_{\hbar}(t=0)=\psi_{\hbar},\end{equation}
and
\begin{equation}\label{e:schrodinger-perturbed}i\hbar\partial_tu_{\hbar}^{\eps}=-\frac{\hbar^2\Delta u_{\hbar}^{\eps}}{2}+\eps_{\hbar}V 
u_{\hbar}^{\eps},\quad u_{\hbar}^{\eps}(t=0)=\psi_{\hbar},\end{equation}
where $V$ belongs to $\ml{C}^{\infty}(\IT^2,\IR)$, $(\psi_{\hbar})_{\hbar\rightarrow 0^+}$ is a normalized sequence in $L^2(\IT^2)$ 
and $(\eps_{\hbar})_{\hbar\rightarrow 0^+}$ satisfies
$$\lim_{\hbar\rightarrow 0^+}\eps_{\hbar}=0.$$
In order to measure the difference between the two evolved systems, Peres used the so-called notion of 
quantum fidelity,
\begin{equation}\label{e:loschmidt-echo}
 \ml{E}_{\hbar,\eps}(t):=\left|\la u_{\hbar}^{\eps}(t), u_{\hbar}(t)\ra_{L^2(\IT^2)}\right|^2.
\end{equation}
 Here, $u_{\hbar}(t)$ represents the solution 
to~\eqref{e:schrodinger-unperturbed} at time $t$ and $u_{\hbar}^{\eps}(t)$ the solution to~\eqref{e:schrodinger-perturbed} at 
time $t$, both of them having the same initial condition which is normalized in $L^2(\IT^2)$. Under this form, this fidelity 
between pure states is now often referred to as the Quantum Loschmidt Echo in the physics literature. Peres predicted that this quantity should decay 
for \emph{any} quantum system and that the rate of decay should depend on the dynamical properties of the underlying classical 
Hamiltonian. He motivated his conjecture with numerical simulations and with the following formal asymptotic expansion:
\begin{equation}\label{e:peres}\ml{E}_{\hbar,\eps}(t)\simeq1-\frac{\eps_{\hbar}^2t^2}{\hbar^2}\left(\la \psi_{\hbar},V^2\psi_{\hbar}\ra -\la \psi_{\hbar}, V\psi_{\hbar}\ra^2\right)
+\ldots.\end{equation}
Hence, for \emph{any} type of Hamiltonian, he deduced that, for short times 
$t\ll\tau_{\hbar}^c:=\frac{\hbar}{\eps_{\hbar}},$
the quantity $\ml{E}_{\hbar,\eps}(t)$ should follow some quadratic decay. This first approximation just relies on the fact that the quantum fidelity 
can be well approximated by the Taylor expansion. After that, the Quantum Loschmidt Echo is expected to continue its 
decay at a rate that will now also depend on the dynamical features of the classical Hamiltonian, namely chaotic vs. integrable. This kind 
of general behaviour seems to be commonly accepted in the literature on the subject -- see for example~\cite[Sect.~2.1.1]{GPSZ06} 
or~\cite[Sect.~2.3.1]{GJPW12}. Note that the rates of decay (after this short time regime) depend in a subtle 
manner on the parameters $\hbar$ and $\eps_{\hbar}$ but also on the choice 
of initial data and of $V$. For much larger time scales, we refer for instance to paragraphs~2.3.1 and 2.3.2 in~\cite{GJPW12} 
for a brief review of the different possible 
regimes and references to the literature -- see also~\cite{JaPe09}. We shall not discuss these questions here and we will mostly focus on the case 
of the transition regime $t\approx\tau_{\hbar}^c$ for which one should already observe some decay of the Quantum Loschmidt Echo. We emphasize that 
we are mainly concerned with the case of the free Schr\"odinger evolution on the $2$-torus which is the simplest example of a \emph{completely integrable} system. 
In this dynamical framework, the situation is known to be quite subtle and many phenomena may occur -- e.g. see~\cite{DuGo14} and the references therein.

One of the main consequences of our 
analysis will be to exhibit large classes of semiclassical initial data for which we do not have any \emph{macroscopic decay} of the Quantum Loschmidt Echo 
in this transition regime -- see section~\ref{s:examples} below. More precisely, for any time $ t\tau_{\hbar}^c$ 
with $t\in\IR$ fixed, the Quantum Loschmidt Echo will tend to $1$ in the semiclassical limit. Note that this does not exclude the possibility that 
$\ml{E}_{\hbar,\eps}(t)$ is strictly less than $1$ and it rather states that the deviation from $1$ is asymptotically small in 
the semiclassical regime without trying to be quantitative in the size of the deviation. Our strategy is to apply to this physical problem 
the $2$-microlocal techniques recently developed by Anantharaman and Maci\`a for the study of controllability and of semiclassical measures 
on flat tori~\cite{Ma10, AM10} -- see also~\cite{AL14, AFM12, ALM16, MaRi16b} for related results on integrable systems. In particular, our analysis 
will be highly dependent on the \emph{integrable structure of our Hamiltonian}.

Note also that our results will be valid for a certain regime of small perturbations\footnote{In all the article, we say that $f_1(\hbar)\ll f_2(\hbar)$ for two sequences 
$f_1(\hbar),f_2(\hbar)> 0$ if $\lim_{\hbar\rightarrow 0^+}f_1(\hbar)f_2(\hbar)^{-1}=0.$}:
\begin{equation}\label{e:perturbation-size}\hbar^2\ll \eps_{\hbar}\ll\hbar.\end{equation}
In fact, the case of stronger perturbations $\eps_{\hbar}\geq \hbar$ can be easily treated for any compact Riemannian 
manifold -- see appendix~\ref{a:strong}. Here, we manage to deal with smaller perturbations due to the specific structure of the torus.
Our results to not a priori extend to other types of 
quantum systems, where different phenomena may occur. Finally, we emphasize that, even if the Quantum Loschmidt Echo is now a rather 
well studied and understood quantity in the physics literature, much less seems to be known from the mathematical perspective. For recent 
mathematical results we refer the reader to~\cite{BolSc06, CoRo07} for $\IR^d$, to~\cite{EsTo12, CanJaTo12} for general 
compact manifolds, to~\cite{EsRi15, Ri16} for negatively curved surfaces and to~\cite{MaRi16} for Zoll manifolds.

\section{Main results}\label{s:main-result}

We will now state our results more precisely. First of all, we emphasize that we will deal with semiclassical sequences of initial data. More precisely, 
in the following, we shall always suppose that $(\psi_{\hbar})_{0<\hbar\leq 1}$ is a family of initial data which is \emph{normalized} in $L^2(\IT^2)$ and which satisfies
\begin{equation}\label{e:shosc}\lim_{\delta\rightarrow 0^+}\limsup_{\hbar\rightarrow 0^+}\left\|\mathbf{1}_{[0,\delta]}(-\hbar^2\Delta)\psi_{\hbar}\right\|_{L^2(\IT^2)}=0,\end{equation}
and
\begin{equation}\label{e:hosc}\lim_{R\rightarrow+\infty}\limsup_{\hbar\rightarrow 0^+}
\left\|\mathbf{1}_{[R,+\infty[}(-\hbar^2\Delta)\psi_{\hbar}\right\|_{L^2(\IT^2)}=0.\end{equation}
Equivalently, the sequence of initial data \emph{oscillates} at the frequency $\hbar^{-1}$. Here, the semiclassical parameter $\hbar\rightarrow 0^+$ is the one appearing in the Schr\"odinger equations~\eqref{e:schrodinger-unperturbed} 
and~\eqref{e:schrodinger-perturbed}. Our main result is to 
establish an explicit formula for the Quantum Loschmidt Echo at the critical time scale
\begin{equation}\label{e:crit-time-scale}
 \tau_{\hbar}^c:=\frac{\hbar}{\eps_{\hbar}},
\end{equation}
in terms of the initial conditions $(\psi_{\hbar})_{0<\hbar\leq 1}$. 
Due to~\eqref{e:perturbation-size}, this time scale tends to $+\infty$ in the semiclassical limit and it is always much smaller than 
the Heisenberg time $\hbar^{-1}$ from the physics literature. In order to state our main result, we need to fix some conventions. 
 We denote by $\ml{L}_1$ the family of all primitive rank $1$ sublattices of $\IZ^2$. Recall that a sublattice $\Lambda$ is said to be primitive if $\la\Lambda\ra\cap\IZ^2=\Lambda$, 
 where $\la\Lambda\ra$ is the subspace of $\IR^2$ spanned by $\Lambda$. Moreover, it is of rank $1$ if $\la \Lambda\ra$ is one dimensional. Any such lattice is generated by an element $\vec{v}_{\Lambda}$ 
 of $\IZ^2$ such that $\IZ \vec{v}_{\Lambda}=\Lambda$. Denote then by $\vec{v}_{\Lambda}^{\perp}$ the lattice vector which is directly orthogonal to $\vec{v}_{\Lambda}$ and which has the same length 
 $L_{\Lambda}:=\|\vec{v}_{\Lambda}\|=\|\vec{v}_{\Lambda}^{\perp}\|$. We then introduce two Hamiltonian functions associated with $\Lambda$: 
 $$\forall \xi\in\IR^2,\ H_{\Lambda}(\xi):=\frac{1}{L_{\Lambda}}\la \xi,\vec{v}_{\Lambda}\ra\ \text{and}\ H_{\Lambda}^{\perp}(\xi):=\frac{1}{L_{\Lambda}}\la \xi,\vec{v}_{\Lambda}^{\perp}\ra.$$
This defines a completely integrable system and the flow corresponding to $H_{\Lambda}^{\perp}$ is defined by
\begin{equation}\label{e:Lambda-flow}\varphi_{H_{\Lambda}^{\perp}}^t(x,\xi):=\left(x+t\frac{\vec{v}_{\Lambda}^{\perp}}{L_{\Lambda}},\xi\right).\end{equation}
Note that this flow is $L_{\Lambda}$-periodic. Introduce then, for every smooth function $b$ on $T^*\IT^2$, 
$$\ml{I}_{\Lambda}(b):=\frac{1}{ L_{\Lambda}}\int_0^{ L_{\Lambda}}b\circ\varphi_{H_{\Lambda}^{\perp}}^t(x,\xi)dt,$$
which has only Fourier coefficients in the direction of $\Lambda$. Using these conventions, we can define the following map
$$\mathbf{F}_{\Lambda,\hbar}:a\in\ml{C}_c^{\infty}(\IT^2\times\IR)\mapsto\left\la\psi_{\hbar},
\Oph\left(\ml{I}_{\Lambda}(a)\left(x,\frac{\hbar H_{\Lambda}(\xi)}{\eps_{\hbar}}\right)\right)\psi_{\hbar}\right\ra,$$
where $\Oph$ is the standard quantization -- see appendix~\ref{a:sc-an}. Roughly speaking, these quantities measure the 
concentration of the initial data in an $\frac{\eps_{\hbar}}{\hbar}$-neighborhood of $\Lambda^{\perp}:=\IR \vec{v}_{\Lambda}^{\perp}$. From our assumption~\eqref{e:perturbation-size}, 
the quantity $\frac{\eps_{\hbar}}{\hbar}$ 
goes to $0$, but not faster than $\hbar$. Up to an extraction, we can suppose that, for every $\Lambda\in\ml{L}_{1}$, there exists 
a finite positive measure\footnote{We refer to paragraph~\ref{s:twomicrolocal} for further precisions on the regularity of these objects.} 
$\mathbf{F}^0_{\Lambda}$ such that, for every $a\in\ml{C}^{\infty}_c(\IT^2\times\IR)$,
\begin{equation}\label{e:2micro-initialdata}\lim_{\hbar\rightarrow 0^+}\la \mathbf{F}_{\Lambda,\hbar},a\ra= \la \mathbf{F}^0_{\Lambda},a\ra.\end{equation}
Hence, the measure $\mathbf{F}^0_{\Lambda}$ describes the part of the mass which is asymptotically concentrated along $\Lambda^{\perp}$. These are rescaled 
versions of the so-called semiclassical measures~\cite{Ge91, Zw12}. It turns out that the asymptotic properties of the Quantum Loschmidt Echo are in fact 
related to these quantities:
\begin{theo}\label{t:evolution} Suppose that~\eqref{e:perturbation-size} holds, and that we are 
given a sequence of normalized initial data $(\psi_{\hbar})_{\hbar\rightarrow 0^+}$ satisfying~\eqref{e:shosc} 
and~\eqref{e:hosc} which generates a unique family $(\mathbf{F}^0_{\Lambda})_{\Lambda\in\ml{L}_1}$.
 \begin{eqnarray*}
\lim_{\hbar\rightarrow 0^+}\left\la u_{\hbar}^{\eps}(t\tau_{\hbar}^c), u_{\hbar}(t\tau_{\hbar}^c)\right\ra &=&e^{it\int_{\IT^2} V}\left(1-\sum_{\Lambda\in\ml{L}_1}\la \mathbf{F}_{\Lambda}^0,1\ra\right)\\
&+ &\sum_{\Lambda\in\ml{L}_1}\int_{\IT^2\times\IR}
e^{i\int_0^t\ml{I}_{\Lambda}(V)\left(x+s\frac{\eta \vec{v}_{\Lambda}}{L_{\Lambda}}\right)ds}\mathbf{F}_{\Lambda}^0(dx,d\eta), 
\end{eqnarray*}
\end{theo}
This Theorem provides an explicit formula for the Quantum Loschmidt Echo at the transition regime $\tau_{\hbar}^c$. It is in some sense slightly 
more precise as we compute the overlap $\left\la u_{\hbar}^{\eps}(t\tau_{\hbar}^c), u_{\hbar}(t\tau_{\hbar}^c)\right\ra$ and not only its modulus. 
If $\mathbf{F}_{\Lambda}^0$ identically vanishes for any $\Lambda$ in $\ml{L}_1$, then this Theorem shows
$$\lim_{\hbar\rightarrow 0^+}\left|\left\la u_{\hbar}^{\eps}(t\tau_{\hbar}^c), u_{\hbar}(t\tau_{\hbar}^c)\right\ra\right|^2=1.$$
The assumption $\mathbf{F}_{\Lambda}^0\equiv 0$ exactly means that the initial data do not concentrate too fast near the invariant tori where the geodesic 
flow is periodic. Again, this does not mean that the Quantum Loschmidt Echo does not decay but it rather states that the deviation from $1$ is asymptotically small. 
In section~\ref{s:examples}, we will apply this result to standard families of initial data such as coherent states 
and plane waves. In that manner, we will illustrate the many possibilities for the behaviour of the Quantum Loschmidt Echo for integrable systems.

Note that a similar statement was obtained by Maci\`a and the first author in the case of \emph{degenerate} integrable systems, like the geodesic flow on the 
sphere~\cite[Sect.~5]{MaRi16}. In that framework, the case of smaller perturbations could be treated up to the scales $\eps_{\hbar}\gg\hbar^3$. 
Compared with that reference, the situation here is more complicated from a dynamical point of view 
as there are directions where the Hamiltonian flow is periodic (with a period that depends on the direction) and other ones where the 
geodesic flow fills the torus. In order to deal with these different behaviours, the main additional ingredient compared 
with~\cite{MaRi16} will be to introduce two microlocal objects as in~\cite{Ma10, AM10, AFM12}, namely $(\mathbf{F}^0_{\Lambda})_{\Lambda\in\ml{L}_1}$. 
These quantities will capture the properties of the Schr\"odinger evolution near directions where the classical Hamiltonian flow is periodic while, away 
from these directions, we will be able to use equidistribution of the geodesic flow. Finally, we emphasize 
that the regime of perturbations we consider here is the same as in~\cite{MaRi16b} which 
describes the structure of semiclassical measures for the Schr\"odinger equation~\eqref{e:schrodinger-perturbed}. However, the two 
problems require in fact a second microlocalization at different scales and they yield propagation laws given by different two-microlocal quantities.

\subsection*{Organization of the article}
 In section~\ref{s:examples}, we start by applying Theorem~\ref{t:evolution} to families of relevant initial data. In 
 section~\ref{s:fidelity-distrib}, we introduce families of distributions on the cotangent bundle $T^*\IT^2$ that are 
 close to the so-called Wigner distributions. Yet, as they are of slightly different nature, we review some of their 
 basic properties such as invariance under the geodesic flow following the classical arguments from~\cite{Ge91, Ma09, Zw12}. 
 We also relate them to the Quantum Loschmidt Echo at the 
 critical time scale, and we reduce the problem to an analysis of their restriction to rational directions. Then, in 
 section~\ref{s:twomicrolocal}, we define the two microlocal framework needed to analyze the behaviour along rational 
 directions. The proof of Theorem~\ref{t:evolution} is given in section~\ref{s:proof}. Finally, the article contains two appendices. 
 Appendix~\ref{a:strong} is devoted to the simpler case of strong perturbations $\eps_{\hbar}\geq\hbar$, while appendix~\ref{a:sc-an} 
 provides a short toolbox of semiclassical analysis on $\IT^2$.

All along the article, we use, for every $u$ and $v$ in $L^2(\IT^2)$,
$$\la u,v\ra_{L^2(\IT^2)}:=\int_{\IT^2}\overline{u}(x)v(x)dx.$$

\subsection*{Acknowledgements} Part of this work was carried out when the second author was a postdoc of the Labex 
CEMPI program (ANR-11-LABX-0007-01). The first author is also partially supported 
by the Agence Nationale de la Recherche through the Labex CEMPI and the 
ANR project GERASIC (ANR-13-BS01-0007-01). The authors warmly thank R\'emi Dubertrand for an interesting discussion related to 
the Quantum Loschmidt Echo for integrable systems. The first author is also grateful to Fabricio Maci\`a 
for many explanations about his works~\cite{Ma10, AM10, AFM12}.

\section{Examples of initial data}\label{s:examples}
The purpose of this section is to apply our results to specific examples of initial data that are frequently discussed in the physics literature: 
plane waves and coherent states (or superposition of such states). The goal of these examples is to illustrate the variety 
of behaviour that may occur. For most of the cases, we shall verify that the Quantum Loschmidt Echo is in fact asymptotically equal to $1$ in the 
transition regime. Recall that one expects a decay for any type of quantum system. Still, in certain 
specific cases, where the quantum state concentrates along periodic trajectories in phase space, 
we may observe different phenomena, e.g. polynomial decay for certain families of plane waves or revivals for a superposition of 
coherent states.

In order to describe these various phenomena, we have to calculate the distributions $\mathbf{F}_\Lambda^0$ for all primitive rank $1$ sublattices $\Lambda$. 
Then, we shall apply the time evolution formula given by Theorem \ref{t:evolution} in order to derive explicit expressions for the Quantum Loschmidt Echo 
for these sequences of initial data. It turns out that these distributions will be nontrivial only if the corresponding sequence of wave vectors ,
when projected on $\IS^1$, converges to a rational direction at a certain rate. In other words, the sequence of initial data is concentrated near the tori of 
periodic orbits of the unperturbed Hamiltonian flow. In some sense, our series of examples illustrates that, for most initial data, we do not observe 
a macroscopic decay for times of order $\tau_{\hbar}^c$. 

In order to state our results, we define the set of ``rational'' unit vectors
as $$\IS^1_\mathbb{Q}:=\{\vec{v}_{\Lambda}/L_\Lambda\mid \Lambda\in\ml{L}_1\}.$$
Again, this corresponds to the directions where the classical Hamiltonian flow is periodic. We also 
introduce the convenient notation $\vec{v}_{\Lambda}/L_\Lambda=
(\cos\alpha_\Lambda,\sin\alpha_\Lambda)$ for some unique $\alpha_\Lambda\in[0,2\pi)$.

\subsection{Plane waves}
Let us consider the initial data $\psi_\hbar(x)=e^{2\pi i kx}=:e_k(x)$ for a sequence of lattice vectors 
$k\in\IZ^2$, $\|k\|\to\infty$. In that case, we choose $$\hbar=\hbar_k=\|k\|^{-1}.$$ 
In this entire paragraph, we suppose that there exists $\vec{v}\in\IS^{1}$
such that $k/\|k\|\to \vec{v}$ as $\|k\|\to\infty$. According to Appendix~\ref{a:sc-an}, one has, for 
$a\in C^\infty_c(\IT^2\times\IR)$ and for $\Lambda\in\ml{L}_1$,
\begin{equation}\label{e:sc-measure-planewave}\left\langle \mathbf{F}_{\Lambda,\hbar},a \right\rangle=
\left\langle \psi_\hbar,\Oph(\ml{I}_{\Lambda}(a)(x,2\pi\hbar\epsilon_\hbar^{-1}H_\Lambda(\xi)))\psi_\hbar\right\rangle=
\widehat{a}(0,2\pi \hbar^2\epsilon_\hbar^{-1}H_\Lambda (k)).\end{equation}

\subsubsection{Rational directions} We start with the case, where the limit vector $\vec{v}$ belongs to $\IS^1_\mathbb{Q}$. In that case, the following holds:
\begin{prop}\label{plane2}
Let $(k:=n(\hbar) \vec{v}_{\Lambda_0}+m(\hbar)\vec{v}_{\Lambda_0}^{\perp})_{\hbar\rightarrow 0^+}$ for some $\Lambda_0\in\ml{L}_1$ and 
with $|m(\hbar)|\ll 1/\hbar$. In particular, $\vec{v}=\vec{v}_{\Lambda_0}/L_{\Lambda_0}$. 
Then, the following holds:
\begin{enumerate}
 \item If $2\pi m(\hbar)\hbar^{2}\epsilon_\hbar^{-1}\to\omega\in\IR$, then
 $$\mathbf{F}_{\Lambda}^0(x,\eta)=
\begin{cases}
\delta_{\omega L_{\Lambda_0}}(\eta)\quad\quad\text{if}\; \la\Lambda\ra=\Lambda_0^{\perp},\\
0\qquad\qquad\qquad\text{otherwise.}
\end{cases}$$
 \item If $2\pi|m(\hbar)|\hbar^{2}\epsilon_\hbar^{-1}\to+\infty$, we have $$\forall\Lambda\in\mathcal{L}_1,\ \mathbf{F}_\Lambda^0=0.$$
\end{enumerate}
\end{prop}
\begin{proof}
We have $H_\Lambda (k)=m L_{\Lambda_0}$ if $\la\Lambda\ra=\Lambda_0^{\perp}$. 
Therefore, in that case, the term $\hbar^2\epsilon_\hbar^{-1}H_\Lambda (k)$ appearing in~\eqref{e:sc-measure-planewave} may remain bounded as 
$\hbar\to0$. If $2\pi m\hbar^{2}\epsilon_\hbar^{-1}\to\omega$, we obtain, for every $a\in C^\infty_c(\IT^2\times\IR)$,
$$\left\langle \mathbf{F}_{\Lambda,\hbar},a \right\rangle=\widehat{a}(0,2\pi\hbar^2\epsilon_\hbar^{-1}H_\Lambda (k))
\to\widehat{a}(0,\omega L_{\Lambda_0})
=\int_{\IT^2\times \IR}a(x,\eta)\delta_{\omega L_{\Lambda_0}}(d\eta)dx .$$
On the other hand, if $m\hbar^{2}\epsilon_\hbar^{-1}\to\infty$, then, one finds
$$\left\langle \mathbf{F}_{\Lambda,\hbar},a \right\rangle=\widehat{a}(0,2\pi\hbar^2\epsilon_\hbar^{-1}m L_{\Lambda_0})=0,$$
for $\hbar$ small enough since $a$ is compactly supported. It remains to discuss the case where $\la\Lambda\ra\neq\Lambda_0^{\perp}$. In that case, we have 
$$\hbar^2\epsilon_\hbar^{-1}|H_\Lambda (k)|\gtrsim \hbar^2\eps_{\hbar}^{-1} |n(\hbar)|\to\infty,$$ 
since $\|k\|\sim\hbar^{-1}$, $m(\hbar)=o(\hbar^{-1})$ and $\eps_{\hbar}\ll\hbar.$ Hence, we can apply~\eqref{e:sc-measure-planewave} one more time to conclude. 
\end{proof}

\subsubsection{Irrational directions} Let us now consider the case where $\vec{v}$ belongs to $\IS^1\setminus \IS^1_\mathbb{Q}$:  

\begin{prop}\label{plane3}
Let $k=(m(\hbar),n(\hbar))_{\hbar\rightarrow 0^+}$ be a sequence of lattice points in $\IZ^2$ such that 
$$\vec{v}=(\cos\alpha,\sin\alpha)\notin\IS^1_\mathbb{Q}.$$ 
Then, one has, for all $\Lambda\in\ml{L}_1$,
$$\mathbf{F}_\Lambda^0=0.$$
\end{prop}
\begin{proof}
For a given $\Lambda$, one has 
$$\hbar H_\Lambda (k)=\hbar\left\langle k,\vec{v}_{\Lambda}/L_\Lambda\right\rangle\to(\cos\alpha\cos\alpha_\Lambda+\sin\alpha\sin\alpha_\Lambda).$$ 
Therefore, as $\hbar$ goes to $0$,
$$\frac{\hbar^2}{\epsilon_\hbar}H_\Lambda (k)\sim \frac{\hbar}{\epsilon_\hbar}\cos(\alpha_\Lambda-\alpha),$$
and it follows that $\hbar^2\epsilon_\hbar^{-1}|H_\Lambda (k)|\to+\infty$ because $\hbar/\epsilon_\hbar\to+\infty$ and $\cos(\alpha_\Lambda-\alpha)\neq0$, since $\alpha\notin\IS^1_\mathbb{Q}$. 
Thanks to~\eqref{e:sc-measure-planewave}, this implies $\left\langle \mathbf{F}_{\Lambda,\hbar},a \right\rangle=\widehat{a}(0,2\pi\hbar^2\epsilon_\hbar^{-1}H_\Lambda (k))\to 0$ as $\hbar\to0$. 
Hence $\mathbf{F}_\Lambda^0=0$ for all $\Lambda\in\ml{L}_1$.
\end{proof}

\subsubsection{Superposition of two plane waves} In this last paragraph, we consider a slightly different example. We will verify that
more complicated measures may arise for a superposition of two plane waves. For instance, one has in the rational case:
\begin{prop}\label{two:plane}
Let $(k(\hbar))_{\hbar\rightarrow 0^+}$ and $(l(\hbar))_{\hbar\rightarrow 0^+}$ be two distinct sequences of lattice points of the form 
$$k(\hbar)=n_1(\hbar)\vec{v}_{\Lambda_1}+m_1(\hbar)\vec{v}_{\Lambda_1}^{\perp}\quad\text{and}\quad l(\hbar)=
n_2(\hbar)\vec{v}_{\Lambda_2}+m_2(\hbar)\vec{v}_{\Lambda_2}^{\perp},$$ 
where $|k(\hbar)|=|l(\hbar)|=\hbar^{-1}$, and, for $j=1,2$, $\Lambda_j\in\ml{L}_1$ with $\Lambda_1\neq\Lambda_2$. Suppose also that, for $j\in\{1,2\}$, 
there exists $\omega_j$ in $\IR$ such that\footnote{Note that this implies $m_j(\hbar)=o(\hbar^{-1})$.}
$$\lim_{\hbar\rightarrow 0^+}2\pi m_j(\hbar)\hbar^2\eps_{\hbar}^{-1}=\omega_j.$$
Set
$$\psi_\hbar(x)=\frac{1}{\sqrt{2}}(e_k(x)+e_l(x)).$$ 
Then, the corresponding measure $\mathbf{F}_{\Lambda}^0$ satisfies the following:
$$\mathbf{F}_{\Lambda}^0(x,\eta)=
\begin{cases}
\tfrac{1}{2}\delta_{\omega_1L_{\Lambda_1}}(\eta)\quad\quad\text{if}\; \la\Lambda\ra=\Lambda_1^\perp,\\
\tfrac{1}{2}\delta_{\omega_2L_{\Lambda_2}}(\eta)\quad\quad\;\text{if}\; \la\Lambda\ra=\Lambda_2^\perp,\\
0\qquad\qquad\qquad\text{otherwise.}
\end{cases}$$
\end{prop}
\begin{proof} Using appendix~\ref{a:sc-an}, we find that
\begin{equation}
\begin{split}
\left\langle \mathbf{F}_{\Lambda,\hbar},a \right\rangle=&\sum_{m\in\{l,k\}}\frac{1}{\sqrt{2}}\sum_{n\in\{l,k\}:m-n\in\Lambda}
\frac{1}{\sqrt{2}}\widehat{a}\left(n-m,2\pi
\frac{\hbar^2}{\epsilon_\hbar}H_\Lambda(m)\right)\\
=&\frac{1}{2}\widehat{a}(0,2\pi\hbar^2\epsilon_\hbar^{-1}H_\Lambda(k))+\frac{1}{2}\widehat{a}(0,2\pi\hbar^2\epsilon_\hbar^{-1}H_\Lambda(l))+\ml{O}_N(\hbar^N),
\end{split}
\end{equation}
where we used the fact that $\|k-l\|\sim\hbar^{-1}$ and $|\widehat{a}(k-l,\hbar^2\epsilon_\hbar^{-1}H_\Lambda(k))|\lesssim_N \|k-l\|^{-N}\sim\hbar^N$.
So, if $\la\Lambda\ra=\Lambda_1^\perp$, then, for all $a\in\ml{C}^{\infty}_c(\IT^2\times\IR)$, we have by the same argument as in the proof of Proposition~\ref{plane2}
$$\left\langle \mathbf{F}_{\Lambda,\hbar},a \right\rangle\sim \frac{1}{2}\widehat{a}(0,2\pi\hbar^2\epsilon_\hbar^{-1}H_\Lambda(k))\to 
\frac{1}{2}\widehat{a}(0,\omega_1L_{\Lambda_1})$$ as $\hbar\to0$. Similarly, if $\la\Lambda\ra=\Lambda_2^\perp$, then 
we have $$\left\langle \mathbf{F}_{\Lambda,\hbar},a \right\rangle\to \frac{1}{2}\widehat{a}(0,\omega_2L_{\Lambda_2}).$$
Otherwise, one can make use of the fact that $a$ is compactly supported in $\eta$ to deduce that 
$\left\langle \mathbf{F}_{\Lambda,\hbar},a \right\rangle\to 0$.
\end{proof}

\subsection{Limit of the Quantum Loschmidt Echo for plane waves}
Now that we have computed the limit measures associated with the initial data, we can apply Theorem \ref{t:evolution} in order to 
derive an explicit formula for the Quantum Loschmidt Echo:
\begin{prop} Suppose that~\eqref{e:perturbation-size} is satisfied. Then, for sequences of plane waves, the following hold:
\begin{enumerate}
\item If $(\psi_\hbar)_{\hbar\rightarrow 0^+}$ verifies the assumptions of Proposition~\ref{plane3} or of part (2) of Proposition~\ref{plane2}, then
$$\lim_{\hbar\rightarrow 0^+}|\left\la u_{\hbar}^{\eps}(t\tau_{\hbar}^c), u_{\hbar}(t\tau_{\hbar}^c)\right\ra|^2=1.$$
\item If $(\psi_\hbar)_{\hbar\rightarrow 0^+}$ verifies the assumptions of part (1) of Proposition~\ref{plane2}, then
$$
\lim_{\hbar\rightarrow 0^+}|\left\la u_{\hbar}^{\eps}(t\tau_{\hbar}^c), u_{\hbar}(t\tau_{\hbar}^c)\right\ra|^2
=\left|\int_{\IT^2}e^{i\int_{0}^t\ml{I}_{\Lambda_0^\perp}(V)(x+s\omega \vec{v}_{\Lambda_0}^\perp)ds}dx\right|^2,$$
\item If $(\psi_\hbar)_{\hbar\rightarrow 0^+}$ verifies the assumptions of Proposition~\ref{two:plane}, then
$$
\lim_{\hbar\rightarrow 0^+}|\left\la u_{\hbar}^{\eps}(t\tau_{\hbar}^c), u_{\hbar}(t\tau_{\hbar}^c)\right\ra|^2
=\left|\frac{1}{2}\sum_{j=1}^2\int_{\IT^2}e^{i\int_{0}^t\ml{I}_{\Lambda_j^\perp}(V)(x+s\omega_j \vec{v}_{\Lambda_j}^\perp)ds}dx\right|^2.$$
\end{enumerate}
In all the statements, we used the conventions of Theorem~\ref{t:evolution}.
\end{prop}

The proof of this Proposition is a direct application of Theorem~\ref{t:evolution} combined with the above Propositions which 
computed $\mathbf{F}_{\Lambda}^0$. Note that the last part of the Proposition could be generalized to a finite superposition 
of plane waves by similar calculations. In the case where $\omega_j=0$ (or $\omega=0$ in part (2)), we can observe that the integral 
reduces to
$$\int_{\IT^2}e^{it\ml{I}_{\Lambda^\perp}(V)(x)}dx,$$
which can be viewed as an integral on the $1$-dimensional torus $\mathbb{S}_{\Lambda^{\perp}}=\IR/(L_{\Lambda}\IZ)$. In particular, if 
$\ml{I}_{\Lambda^\perp}(V)$ has only finitely many critical points as a function on the $1$-dimensional torus, then an application of 
stationary phase asymptotics states that this integral decays as $t^{-1/2}$ which yields in some sense a decay of the Quantum Loschmidt Echo. 
On the other hand, plane waves associated with an irrational limit vector $\vec{v}$ do not give rise to any macroscopic decay of the 
Quantum Loschmidt Echo.

\subsection{Coherent states}\label{coherent}
We first define a notion of (generalized) coherent state following 
for instance~\cite{AFM12}. Let $\varphi\in \ml{C}^\infty_c(\IR^2)$ with\footnote{Note that the 
arguments could be extended to deal with the more classical case where 
$\varphi=\exp(-\|x\|^2/2).$} $\supp \varphi\subset(-\tfrac{1}{2},\tfrac{1}{2})^2$, $\|\varphi\|_2=1$. We define a 
coherent state on $\IR^2$ by $$\varphi_\hbar(x)=\hbar^{-1/2}\varphi\left(\frac{x-x_0}{\sqrt\hbar}\right)e^{2\pi i \frac{\xi_0}{\hbar} \cdot x},$$
and we periodize it to obtain a coherent state on the torus, $\psi_\hbar\in \ml{C}^\infty(\IT^2)$,
\begin{equation}\label{e:coherent-torus}\psi_\hbar(x)=\sum_{l\in\IZ^2}\varphi_\hbar(x+l).\end{equation}
The semiclassical measure associated with this sequence of initial data is $\delta_{x_0,\xi_0}$. In the following, we 
make the assumption that $\xi_0\neq 0$ in order to verify property~\eqref{e:shosc}. We can verify that the $L^2$ norm of $\psi_{\hbar}$ is 
asymptotically equal to $1$ -- see the proof of Proposition~\ref{coherent3} for instance.

The Fourier decomposition of $\psi_{\hbar}$ can be easily expressed in terms of the Fourier transform of $\varphi$. Indeed, for $k$ in $\IZ^2$, one has
\begin{equation}\label{e:Fourier-R-periodic}\widehat{\psi_{\hbar}}(k)=\hbar^{1/2}e^{-2\pi i( k-\frac{\xi_0}{\hbar})\cdot x_0}\widehat{\varphi}\left(\hbar^{1/2}\left(k-\frac{\xi_0}{\hbar}\right)
 \right),\end{equation}
where $\widehat{\varphi}$ denotes the Fourier transform on $\IR^2$. 
Recall that $\|\widehat{\varphi}\|_{L^2(\IR^2)}=\|\varphi\|_{L^2(\IR^2)}=1$. 
\begin{rema}\label{r:decay-fourier-coeff-coherent-state} Let us make the following useful observation that we shall use at several stages of our argument. We fix a sequence of radii $r_{\hbar}\gg\hbar^{-1/2}$ and, 
using~\eqref{e:Fourier-R-periodic}, one has
$$\sum_{\|m-\xi_0\hbar^{-1}\|\geq r_{\hbar} }|\widehat{\psi_\hbar}(m)|^2
\leq\hbar\sum_{\|m-\xi_0\hbar^{-1}\|\geq r_{\hbar} }
\left|\widehat{\varphi}\left(\hbar^{1/2}\left(m-\frac{\xi_0}{\hbar}\right)\right)\right|^2.$$
As $\widehat{\varphi}$ is rapidly decaying, there exists $C>0$ such that
$$\left|\widehat{\varphi}\left(\hbar^{1/2}\left(m-\frac{\xi_0}{\hbar}\right)\right)\right|\leq C\left(1+
\left\|\hbar^{1/2}\left(m-\frac{\xi_0}{\hbar}\right)\right\|^2\right)^{-2},$$
from which we can infer
\begin{equation}
\begin{split}\sum_{\|m-\xi_0\hbar^{-1}\|\geq r_{\hbar} }|\widehat{\psi_\hbar}(m)|^2
 & \leq  C\hbar\int_{\|y-\xi_0\hbar^{-1}\|\geq r_{\hbar}}
\left(1+\left\|\hbar^{1/2}\left(y-\frac{\xi_0}{\hbar}\right)\right\|^2\right)^{-2}dy\\
&\leq  C\int_{\|y'\|\geq r_{\hbar}\hbar^{\frac{1}{2}}}
\left(1+\left\|y'\right\|^2\right)^{-2}dy'\rightarrow 0.
\end{split}
\end{equation}
The same calculation shows that 
$$\sum_{\|m-\xi_0\hbar^{-1}\|\geq r_{\hbar}}|\widehat{\psi_\hbar}(m)|=o_{\hbar\to0}(\hbar^{-1/2}).$$
\end{rema}

\subsubsection{A preliminary reduction of $\mathbf{F}_{\Lambda,\hbar}$}
We aim at computing the limit distribution derived from the sequence 
$(\mathbf{F}_{\Lambda,\hbar})_{\hbar\rightarrow 0^+}$. For that purpose, we start with a general computation valid for any regime 
of perturbations. For $a$ in $\ml{C}^{\infty}_c(\IT^2\times\IR)$, write
\begin{eqnarray*}\left\langle \mathbf{F}_{\Lambda,\hbar},a\right\rangle & = & \left\langle\psi_\hbar, \Oph(\ml{I}_{\Lambda}(a)(x,\hbar H_\Lambda(\xi)/\epsilon_\hbar)))\psi_\hbar\right\rangle\\
 & = & \sum_{\substack{m\in\IZ^2 \\ 2\pi\hbar^2\epsilon_\hbar^{-1}H_\Lambda (m)\in\supp_\eta(a)}}\overline{\widehat{\psi_\hbar}(m)}
 \sum_{n\in\IZ^2:m-n\in\Lambda}\widehat{\psi_\hbar}(n)
 \widehat{a}\left(m-n,2\pi\frac{\hbar^2}{\epsilon_\hbar}H_\Lambda (m)\right).\end{eqnarray*}
Let us now compute the sum over $n$ in $\IZ^2$ by observing that  $\psi_\hbar$ vanishes surely outside a ball 
$B(x_0,\hbar^{\frac{3}{8}})$. This implies, as $\hbar\to0$,
\begin{equation}\label{asymp}
\begin{split}
\sum_{n\in\IZ^2:m-n\in\Lambda}\widehat{\psi_\hbar}(n)\widehat{a}\left(m-n,2\pi\frac{\hbar^2}{\epsilon_\hbar}H_\Lambda (m)\right)
=&\int_{\IT^2}\ml{I}_{\Lambda}(a)\left(x,2\pi\frac{\hbar^2}{\epsilon_\hbar}H_\Lambda (m)\right)e_{-m}(x)\overline{\psi_\hbar(x)}dx\\
=&\int_{B(x_0,\hbar^{\frac{3}{8}})}\ml{I}_{\Lambda}(a)\left(x,2\pi\frac{\hbar^2}{\epsilon_\hbar}H_\Lambda (m)\right)e_{-m}(x)\overline{\psi_\hbar(x)}dx\\
=&\; \ml{I}_{\Lambda}(a)\left(x_0,2\pi\frac{\hbar^2}{\epsilon_\hbar}H_\Lambda (m)\right)\widehat{\psi_\hbar}(m)\\
&\quad+\ml{O}_a \left( \hbar^{\frac{3}{8}}\int_{B(x_0,\hbar^{\frac{3}{8}})}|\psi_\hbar(x)|dx\right)\\
=&\; \ml{I}_{\Lambda}(a)\left(x_0,2\pi\frac{\hbar^2}{\epsilon_\hbar}H_\Lambda (m)\right)\widehat{\psi_\hbar}(m)+\ml{O}_a(\hbar^{\frac{3}{4}}),
\end{split}
\end{equation}
where we used the bound 
$$\int_{B(x_0,\hbar^{\frac{3}{8}})}|\psi_\hbar(x)|dx\leq \Vol(B(x_0,\hbar^{\frac{3}{8}}))^{1/2}\|\psi_\hbar\|_2=\ml{O}(\hbar^{\frac{3}{8}}).$$
In particular, we have
\begin{equation}\label{est}
|\left\langle \mathbf{F}_{\Lambda,\hbar},a\right\rangle|\lesssim 
\|a\|_\infty\sum_{\substack{m\in\IZ^2 \\ 2\pi\hbar^2\epsilon_\hbar^{-1}H_\Lambda (m)\in\supp_\eta(a)}}|\widehat{\psi_\hbar}(m)|^2
+\ml{O}_a(\hbar^{\frac{3}{4}})\sum_{\substack{m\in\IZ^2 \\ 2\pi\hbar^2\epsilon_\hbar^{-1}H_\Lambda (m)\in\supp_\eta(a)}}|\widehat{\psi_\hbar}(m)|.
\end{equation}

We are now in a position to compute $\mathbf{F}_{\Lambda}^0$. For that purpose, we shall distinguish different regimes depending on the 
relative size of $\eps_{\hbar}$ and $\hbar$

 \subsubsection{Small perturbations $\hbar^2\ll\epsilon_\hbar\ll\hbar^{3/2}$}  We start with the case of small perturbations. In that case, one has:
 \begin{prop}\label{coherent1} Suppose that $\hbar^2\ll\epsilon_\hbar\ll\hbar^{3/2}$. Then, for the sequence of initial data defined by~\eqref{e:coherent-torus} with 
 $\xi_0\neq 0$ and for every $\Lambda\in\mathcal{L}_1$, we have $\mathbf{F}_{\Lambda}^0=0$.
 \end{prop}
 \begin{proof} We fix $a$ in $\ml{C}^{\infty}_c(\IT^2\times\IR)$ and we want to compute $\la\mathbf{F}_\Lambda^0,a\ra.$
We start with the case where $\xi_0\notin\Lambda^\perp$. Then, 
there exists a constant $0<c_{\Lambda}(\xi_0)\leq 1$ such that 
\begin{equation}\label{e:lower-bound-integers}\inf\left\{\left\|m-\frac{\xi_0}{\hbar}\right\| : m\in\IZ^2,\; 2\pi\hbar^2\epsilon_\hbar^{-1}H_\Lambda (m)\in\supp_\eta(a)\right\}\geq 
\frac{c_{\Lambda}(\xi_0)}{\hbar}.\end{equation}
Indeed, suppose for a contradiction that there exist $\delta_{\hbar}\rightarrow 0$ and a sequence of lattice points $m_{\hbar}\in\IZ^2$ 
such that $\|m_{\hbar}-\frac{\xi_0}{\hbar}\|\leq\delta_{\hbar}\hbar^{-1}$ and such that $2\pi\hbar^2\epsilon_\hbar^{-1}H_\Lambda (m_{\hbar})\in\supp_\eta(a)$. 
Decompose $m_{\hbar}$ as $m_{\hbar}=s_{\hbar}\xi_0+t_{\hbar}\xi_0^\perp$.
It follows that $|s_{\hbar}-\hbar^{-1}|,|t_{\hbar}|\leq\delta_{\hbar}\hbar^{-1}\|\xi_0\|^{-1}$ which implies 
$|s_{\hbar}|\geq(1-\delta_{\hbar}\|\xi_0\|^{-1})\hbar^{-1}$. Now write
$$H_{\Lambda}(m_{\hbar})=s_{\hbar}\left\la  \frac{\vec{v}_{\Lambda}}{L_{\Lambda}},\xi_0\right\ra
+t_{\hbar}\left\la  \frac{\vec{v}_{\Lambda}}{L_{\Lambda}},\xi_0^{\perp}\right\ra=\frac{1}{\hbar}
\left\la  \frac{\vec{v}_{\Lambda}}{L_{\Lambda}},\xi_0\right\ra+o\left(\hbar^{-1}\right).$$
Then, we use that $|H_\Lambda(m_{\hbar})|\lesssim \hbar^{-2}\epsilon_\hbar\ll\hbar^{-1}$, which yields the contradiction as $\xi_0\notin\Lambda^{\perp}$. 
In view of the estimate~\eqref{est}, all that remains to be shown is 
$$\sum_{m:\|m-\xi_0\hbar^{-1}\|\geq c_{\Lambda}(\xi_0)\hbar^{-1}} |\widehat{\psi}_\hbar(m)|^2\to0\quad\text{and}
\quad\sum_{m:\|m-\xi_0\hbar^{-1}\|\geq c_{\Lambda}(\xi_0)\hbar^{-1}} |\widehat{\psi}_\hbar(m)|=o(\hbar^{-\frac{3}{4}}),$$ 
as $\hbar\to0$. This is exactly the content of Remark~\ref{r:decay-fourier-coeff-coherent-state}. It follows that $\mathbf{F}_\Lambda^0=0$ 
for every $\Lambda\in\ml{L}_1$ such that $\xi_0\notin\Lambda^{\perp}$. Note that this part of the argument is in fact valid for any 
$\hbar^2\ll\eps_{\hbar}\ll\hbar$.

On the other hand, suppose that $\xi_0\in\Lambda^\perp$. We shall use~\eqref{est} one more time. Yet, we have to argue in a slightly different manner and to 
make use of the fact that $\eps_{\hbar}\ll\hbar^{\frac{3}{2}}$, equivalently $\hbar^{-\frac{3}{2}}\eps_{\hbar}\rightarrow 0$. 
Using~\eqref{e:Fourier-R-periodic}, one has
$$\sum_{\substack{m\in\IZ^2 \\ 2\pi\hbar^2\epsilon_\hbar^{-1}H_\Lambda (m)\in\supp_\eta(a)}}|\widehat{\psi_\hbar}(m)|^2
\leq\hbar\sum_{\substack{m\in \IZ^2\\ |H_{\Lambda}(m)|\leq C_a\hbar^{-2}\eps_{\hbar}}}
\left|\widehat{\varphi}\left(\hbar^{1/2}\left(m-\frac{\xi_0}{\hbar}\right)\right)\right|^2,$$
for some $C_a>0$ that depends only on $a$. Recall from our assumptions that $\hbar^{-2}\eps_{\hbar}\rightarrow+\infty$. As $\widehat{\varphi}$ is 
rapidly decaying, one knows that there exists $C>0$ such that
$$\left|\widehat{\varphi}\left(\hbar^{1/2}\left(m-\frac{\xi_0}{\hbar}\right)\right)\right|\leq C\left(1+
\left\|\hbar^{1/2}\left(m-\frac{\xi_0}{\hbar}\right)\right\|^2\right)^{-1}.$$
Hence, implementing this bound in~\eqref{est} and using that $\xi_0\in\Lambda^{\perp}$ yields
\begin{equation}
\begin{split}\sum_{\substack{m\in\IZ^2 \\ 2\pi\hbar^2\epsilon_\hbar^{-1}H_\Lambda (m)\in\supp_\eta(a)}}|\widehat{\psi_\hbar}(m)|^2
 & \leq  C^2\hbar\int_{|H_{\Lambda}(y-\xi_0\hbar^{-1})|\leq C_a\hbar^{-2}\eps_{\hbar}}
\left(1+\left\|\hbar^{1/2}\left(y-\frac{\xi_0}{\hbar}\right)\right\|^2\right)^{-2}dy\\
&\leq  C^2\int_{|H_{\Lambda}(y')|\leq C_a\hbar^{-\frac{3}{2}}\eps_{\hbar}}
\left(1+\left\|y'\right\|^2\right)^{-2}dy'\rightarrow 0,\ \text{as}\ \hbar\rightarrow 0^+.
\end{split}
\end{equation}
A similar calculation shows that 
$$\sum_{\substack{m\in\IZ^2 \\ 2\pi\hbar^2\epsilon_\hbar^{-1}H_\Lambda (m)\in\supp_\eta(a)}}|\widehat{\psi_\hbar}(m)|=\ml{O}(\hbar^{-1/2}),\ \text{as}\ \hbar\rightarrow 0^+.$$
Indeed, equation~\eqref{e:Fourier-R-periodic} yields
\begin{equation}\label{l1_bound}
\begin{split}
\sum_{\substack{m\in\IZ^2 \\ \hbar^2\epsilon_\hbar^{-1}H_\Lambda (m)\in\supp_\eta(a)}}|\widehat{\psi_\hbar}(m)|
&\leq \hbar^{1/2}\sum_{\substack{m\in \IZ^2\\ |H_{\Lambda}(m-\xi_0\hbar^{-1})|\leq C_a\hbar^{-2}\eps_{\hbar}}}
\left|\widehat{\varphi}\left(\hbar^{1/2}\left(m-\frac{\xi_0}{\hbar}\right)\right)\right|\\
 & \leq  C'\hbar^{1/2}\int_{|H_{\Lambda}(y-\xi_0\hbar^{-1})|\leq C_a\hbar^{-2}\eps_{\hbar}}
\left(1+\left\|\hbar^{1/2}\left(y-\frac{\xi_0}{\hbar}\right)\right\|^2\right)^{-2}dy\\
&\leq  C'\hbar^{-1/2}\int_{|H_{\Lambda}(y')|\leq C_a\hbar^{-\frac{3}{2}}\eps_{\hbar}}
\left(1+\left\|y'\right\|^2\right)^{-2}dy'\\
& \leq  C'\hbar^{-1/2}\int_{\IR^2}
\left(1+\left\|y'\right\|^2\right)^{-2}dy'\\
&=\ml{O}(\hbar^{-1/2}), \;\text{as}\ \hbar\rightarrow 0^+.
\end{split}
\end{equation}
Gathering these upper bounds implies that the limit distribution $\mathbf{F}_\Lambda^0$ is zero for any $\Lambda$. Along the way, we also note that the upper 
bound~\eqref{l1_bound} did not rely on the fact that $\eps_{\hbar}\ll\hbar^{\frac{3}{2}}$. 
\end{proof} 

\subsubsection{Critical perturbations $\epsilon_\hbar=\hbar^{3/2}$}
\begin{prop}\label{coherent3}
Suppose that $\epsilon_\hbar= \hbar^{3/2}$ and that the sequence of initial data is given by~\eqref{e:coherent-torus} with 
$\xi_0\neq 0$.  Then, we have
$$\mathbf{F}_\Lambda^0(x,\eta)=
\begin{cases}
\delta_{x_0}(x)\mu(\eta)\quad\text{if $\xi_0\in\Lambda^\perp$,}\\
\\
0\quad \text{otherwise.}
\end{cases}$$
where $$\int_{\IR}a(\eta)\mu(d\eta)=\int_{\IR^2}a(2\pi H_{\Lambda}(\xi))|\widehat{\varphi}(\xi)|^2d\xi.$$
\end{prop}
\begin{proof}The first part of the proof of Proposition~\ref{coherent1} applies to the case where $\xi_0\notin\Lambda^{\perp}$ 
and $\eps_{\hbar}\geq\hbar^{\frac{3}{2}}$. Hence, it remains to treat the case 
where $\xi_0\in\Lambda^{\perp}.$ In that case, equation~\eqref{asymp} yields
$$\left\langle \mathbf{F}_{\Lambda,\hbar},a\right\rangle=\sum_{\substack{m\in\IZ^2 \\ 2\pi\hbar^{\frac{1}{2}}H_\Lambda (m)\in\supp_\eta(a)}}
\ml{I}_{\Lambda}(a)\left(x_0,2\pi\hbar^{\frac{1}{2}}H_\Lambda (m)\right)|\widehat{\psi_\hbar}(m)|^2+\ml{O}_{a}(\hbar^{\frac{3}{4}})|\widehat{\psi_\hbar}(m)|.$$
The argument of Remark~\ref{r:decay-fourier-coeff-coherent-state} shows that
$$\lim_{R\rightarrow+\infty}\limsup_{\hbar\rightarrow 0^+}\sum_{\|m-\xi_0\hbar^{-1}\|\geq R\hbar^{-1/2} }|\widehat{\psi_\hbar}(m)|^2=0.$$
Thus,
$$\left\langle \mathbf{F}_{\Lambda,\hbar},a\right\rangle=\sum_{\|m-\xi_0\hbar^{-1}\|\leq R\hbar^{-1/2}}
\left(\ml{I}_{\Lambda}(a)\left(x_0,2\pi\hbar^{\frac{1}{2}}H_\Lambda (m)\right)|\widehat{\psi_\hbar}(m)|^2+\ml{O}_{a}(\hbar^{\frac{3}{4}})|\widehat{\psi_\hbar}(m)|\right)+r(R,\hbar),$$
where $\lim_{R\rightarrow+\infty}\limsup_{\hbar\rightarrow 0^+}r(R,\hbar)=0$. We start by estimating the first part of the sum which is equal to
$$\hbar\sum_{\|m-\xi_0\hbar^{-1}\|\leq R\hbar^{-1/2}}
\ml{I}_{\Lambda}(a)\left(x_0,2 \pi\hbar^{\frac{1}{2}}H_\Lambda \left(m-\frac{\xi_0}{\hbar}\right)\right)\left|\widehat{\varphi}\left(\hbar^{1/2}\left(m-\frac{\xi_0}{\hbar}\right)\right)\right|^2,$$
where we used that $\xi_0\in\Lambda^{\perp}$. Letting $\hbar\rightarrow 0^+$, one finds
\begin{equation}
\begin{split}
\left\langle \mathbf{F}_{\Lambda,\hbar},a\right\rangle=&\int_{\|\xi\|\leq R}|\widehat{\varphi}(\xi)|^2 \ml{I}_\Lambda(a)(x_0,2\pi H_{\Lambda}(\xi))d\xi+
\ml{O}_{a}(\hbar^{\frac{3}{4}})\sum_{\substack{m\in\IZ^2 \\ 2\pi\hbar^{\frac{1}{2}}H_\Lambda (m)\in\supp_\eta(a)}}|\widehat{\psi_\hbar}(m)|+r(R,\hbar)\\
=&\int_{\|\xi\|\leq R}|\widehat{\varphi}(\xi)|^2 \ml{I}_\Lambda(a)(x_0,2\pi H_{\Lambda}(\xi))d\xi+\ml{O}_{a}(\hbar^{\frac{1}{4}})+r(R,\hbar),
\end{split}
\end{equation}
where we used estimate \eqref{l1_bound} in the last line. The result follows by taking $\hbar\to 0^+$ and then $R\to+\infty$. As was already mentionned, 
we note that this argument also shows that $\|\psi_{\hbar}\|_{L^2}$ is asymptotically equal to $1$.
\end{proof}

\subsubsection{Large perturbations $\hbar^{3/2}\ll\epsilon_\hbar\ll \hbar$}
In this regime, we have $\hbar^{-1/2}\ll\epsilon_\hbar \hbar^{-2}\ll \hbar^{-1}$ and the following holds:

\begin{prop}\label{coherent2}
Suppose that $\hbar^{3/2}\ll\epsilon_\hbar\ll \hbar$ and that the sequence of initial data is given by~\eqref{e:coherent-torus} with 
$\xi_0\neq 0$.  Then, we have
$$\mathbf{F}_\Lambda^0(x,\eta)=
\begin{cases}
\delta_{x_0}(x)\delta_0(\eta)\quad\text{if $\xi_0\in\Lambda^\perp$},\\
\\
0\quad \text{otherwise.}
\end{cases}$$
\end{prop}
In particular, as it was already the case in Proposition~\ref{coherent3}, if $\xi_0/\|\xi_0\|\notin\mathbb{S}^1_{\mathbb{Q}}$, then $\mathbf{F}_{\Lambda}^0=0$ for every $\Lambda\in\ml{L}_1$.

\begin{proof} 
Arguing as in the proof of Proposition~\ref{coherent1}, we can verify that $\mathbf{F}_\Lambda^0=0$ for every 
$\Lambda\in\ml{L}_1$ such that $\xi_0\notin\Lambda^{\perp}$. Hence, as above, we just need to 
discuss the case where $\xi_0\in\Lambda^{\perp}$. From~\eqref{asymp}, one knows that
$$\left\langle \mathbf{F}_{\Lambda,\hbar},a\right\rangle=\sum_{\substack{m\in\IZ^2 \\ 2\pi\hbar^2\epsilon_\hbar^{-1}H_\Lambda (m)\in\supp_\eta(a)}}
\ml{I}_{\Lambda}(a)\left(x_0,2\pi\frac{\hbar^2}{\epsilon_\hbar}H_\Lambda (m)\right)|\widehat{\psi_\hbar}(m)|^2+\ml{O}_{a}(\hbar^{\frac{3}{4}})|\widehat{\psi_\hbar}(m)|.$$
We fix a sequence of radii $r_{\hbar}$ such that $\epsilon_\hbar\hbar^{-2}\gg r_\hbar\gg\hbar^{-1/2}$. The argument of 
Remark~\ref{r:decay-fourier-coeff-coherent-state} 
allows to show that
\begin{eqnarray*}\left\langle \mathbf{F}_{\Lambda,\hbar},a\right\rangle & = &\sum_{m:\|m-\xi_0\hbar^{-1}\|\leq r_{\hbar}}
\ml{I}_{\Lambda}(a)\left(x_0,2\pi\frac{\hbar^2}{\epsilon_\hbar}H_\Lambda (m)\right)|\widehat{\psi_\hbar}(m)|^2\\ 
& +& \ml{O}_{a}(\hbar^{\frac{3}{4}})\sum_{\substack{m\in\IZ^2 \\ 2\pi\hbar^2\epsilon_\hbar^{-1}H_\Lambda (m)\in\supp_\eta(a)}}
|\widehat{\psi_\hbar}(m)|+o(1).\end{eqnarray*}
Hence, using Remark~\ref{r:decay-fourier-coeff-coherent-state} one more time, the fact that $\|\psi_{\hbar}\|\rightarrow 1$ and the fact that $\xi_0\in\Lambda^{\perp}$, one has
\begin{equation}
\begin{split}
\left\langle \mathbf{F}_{\Lambda,\hbar},a\right\rangle = &\ml{I}_{\Lambda}(a)\left(x_0,0\right)
+ \ml{O}_{a}(\hbar^{\frac{3}{4}})\sum_{\substack{m\in\IZ^2 \\ 2\pi\hbar^2\epsilon_\hbar^{-1}H_\Lambda (m)\in\supp_\eta(a)}}
|\widehat{\psi_\hbar}(m)|+o(1)\\
= &\ml{I}_{\Lambda}(a)\left(x_0,0\right)+o(1),
\end{split}
\end{equation}
where we used estimate \eqref{l1_bound} in the last line.
\end{proof}

\subsubsection{Superpositions of two coherent states}\label{two_coherent}
Let us now consider the initial data $$\psi_\hbar=\frac{1}{\sqrt{2}}(\psi_\hbar^{(x_0,\xi_0)}+\psi_\hbar^{(y_0,\eta_0)}),$$ where $\psi_\hbar^{(x_0,\xi_0)}$ (resp. $\psi_\hbar^{(y_0,\eta_0)}$) 
denotes a coherent state centered at $(x_0,\xi_0)$ (resp. $(y_0,\eta_0)$) with $\xi_0,\eta_0\neq 0$. We will also suppose for the sake of simplicity that 
$x_0\neq y_0$. The case $x_0=y_0$ could be treated in a similar manner but it would require slightly more work. As the case $x_0\neq y_0$ already displays interesting 
features regarding the question of the Quantum Loschmidt Echo, we limit ourselves to this case. The key observation is that we have
\begin{equation}\label{e:diag-approx}
\left\langle \mathbf{F}_{\Lambda,\hbar},a \right\rangle=
\frac{1}{2}\left\langle \mathbf{F}_{\Lambda,\hbar}^{(x_0,\xi_0)},a \right\rangle+
\frac{1}{2}\left\langle \mathbf{F}_{\Lambda,\hbar}^{(y_0,\eta_0)},a \right\rangle+o(1). 
\end{equation}
Hence, the calculation reduces to the analysis of a single coherent state as it was performed in the above Propositions. To see this, 
we simply have to show that the off-diagonal terms are small, i.e. as $\hbar\to 0^+$: 
$$r_\hbar^{\Lambda}:=\left\langle\psi_\hbar^{(x_0,\xi_0)}, 
\Oph(\ml{I}_{\Lambda}(a)(x,\hbar H_\Lambda(\xi)/\epsilon_\hbar))\psi_\hbar^{(y_0,\eta_0)}\right\rangle=o(1).$$
For that purpose, we write that
$$r_{\hbar}^{\Lambda}=\sum_{k\in\IZ^2}e^{2i\pi\frac{\xi_0.k}{\hbar}}
\left\la\varphi_{\hbar}^{x_0+k,\xi_0},\Oph\left(\ml{I}_{\Lambda}(a)(x,\hbar\eps_{\hbar}^{-1}H_{\Lambda}(\xi))\right)\varphi_{\hbar}^{y_0,\eta_0}\right\ra_{L^2(\IR^2)}.$$
We will now estimate each term in the sum by making use of the fact that $x_0\neq y_0$:
\begin{eqnarray*}
 r_{\hbar}^{\Lambda}(k)& := &\left\la\varphi_{\hbar}^{x_0+k,\xi_0},
 \Oph\left(\ml{I}_{\Lambda}(a)(x,\hbar\eps_{\hbar}^{-1}H_{\Lambda}(\xi))\right)\varphi_{\hbar}^{y_0,\eta_0}\right\ra_{L^2(\IR^2)}\\
  & = &\frac{1}{(2\pi\hbar)^2}\int_{\IR^6}e^{\frac{i}{\hbar}\la x-y,\xi\ra}\overline{\varphi_{\hbar}^{x_0+k,\xi_0}(x)}\ml{I}_{\Lambda}(a)(x,\hbar\eps_{\hbar}^{-1}H_{\Lambda}(\xi))\varphi_{\hbar}^{y_0,\eta_0}(y)dxdyd\xi\\
& = &\frac{e^{i\alpha(\hbar)}}{(2\pi)^2}\int_{\IR^6}e^{\frac{i}{\sqrt{\hbar}}\theta_{x_0,\xi_0,y_0,\eta_0}^{(k)}(x,y,\xi)} e^{\frac{i}{\hbar}\la x_0+k-y_0,\xi\ra}
\overline{\varphi(x)}\ml{I}_{\Lambda}(a)(x_0+\sqrt{\hbar}x,\hbar\eps_{\hbar}^{-1}H_{\Lambda}(\xi))\varphi(y)dxdyd\xi,
 \end{eqnarray*}
where $\alpha(\hbar)$ is some real number depending on $(x_0,\xi_0,y_0,\eta_0)$, and where
$$\theta_{x_0,\xi_0,y_0,\eta_0}^{(k)}(x,y,\xi):=-\la x,2\pi\xi_0\ra+\la y,2\pi\eta_0\ra+\la x-y,\xi\ra.$$
Note that we have identified $x_0$ and $y_0$ with elements in $[0,1)^2$. We can now use the fact that
$$\frac{\hbar(x_0+k-y_0).\partial_{\xi}}{i\|x_0+k-y_0\|^2}\left(e^{\frac{i}{\hbar}\la x_0+k-y_0,\xi\ra}\right)=e^{\frac{i}{\hbar}\la x_0+k-y_0,\xi\ra},$$
and integrate by parts. In that manner, we find that, for every $N\geq 1$ and for every $k$ in $\IZ^2$, 
$$r_{\hbar}^{\Lambda}(k)=\|x_0+k-y_0\|^{-N}\ml{O}_N\left(\hbar^{\frac{N}{2}}+(\hbar^2\eps_{\hbar}^{-1})^N\right),$$
which allows to conclude.

\subsection{Limit of the Quantum Loschmidt Echo for coherent states}
Combining Theorem \ref{t:evolution} with the above Propositions yields the following estimates for the evolution of the Quantum Loschmidt Echo:
\begin{prop} Suppose that~\eqref{e:perturbation-size} is satisfied and that the sequence of initial data is given by~\eqref{e:coherent-torus} with 
$\xi_0\neq 0$. Then, the following holds:
\begin{enumerate}
\item If $\eps_{\hbar}\gg\hbar^{\frac{3}{2}}$ or $\eps_{\hbar}\ll\hbar^{\frac{3}{2}}$, then
$$\lim_{\hbar\rightarrow 0^+}|\left\la u_{\hbar}^{\eps}(t\tau_{\hbar}^c), u_{\hbar}(t\tau_{\hbar}^c)\right\ra|^2=1.$$
\item If $\eps_{\hbar}=\hbar^{\frac{3}{2}}$ and $\xi_0/\|\xi_0\|\notin\IS^1_{\mathbb{Q}}$, then
$$
\lim_{\hbar\rightarrow 0^+}|\left\la u_{\hbar}^{\eps}(t\tau_{\hbar}^c), u_{\hbar}(t\tau_{\hbar}^c)\right\ra|^2
=1.$$
\item If $\eps_{\hbar}=\hbar^{\frac{3}{2}}$ and $\xi_0\in\Lambda^{\perp}$ for some $\Lambda\in\ml{L}_1$, then
$$
\lim_{\hbar\rightarrow 0^+}|\left\la u_{\hbar}^{\eps}(t\tau_{\hbar}^c), u_{\hbar}(t\tau_{\hbar}^c)\right\ra|^2
=\left|\int_{\IR^2}|\widehat{\varphi}(\xi)|^2 e^{i\int_0^t\ml{I}_{\Lambda}(V)\left(x_0+2\pi sH_{\Lambda}(\xi)\right)ds}d\xi\right|^2.$$
\end{enumerate}
In all the statements, we used the conventions of Theorem~\ref{t:evolution}.
\end{prop}
In particular, this Proposition shows that, for most of the cases, we do not have any macroscopic decay of the Quantum Loschmidt Echo when the initial data are given by a sequence of coherent states. In the case where 
$\eps_{\hbar}\gg\hbar^{\frac{3}{2}}$, it is interesting to mention the case of a superposition of two coherent states pointing along rational directions:
\begin{prop} Suppose that $\hbar^{\frac{3}{2}}\ll\eps_{\hbar}\ll\hbar$ and that we are given the 
sequence of initial data from paragraph~\ref{two_coherent} with $\xi_0\in\Lambda_1^{\perp}$ and $\eta_0\in\Lambda_2^{\perp}$. Then, the following holds:
$$
\lim_{\hbar\rightarrow 0^+}|\left\la u_{\hbar}^{\eps}(t\tau_{\hbar}^c), u_{\hbar}(t\tau_{\hbar}^c)\right\ra|^2
=\left|\cos\left(\frac{t(\ml{I}_{\Lambda_1}(V)(x_0)-\ml{I}_{\Lambda_2}(V)(y_0))}{2}\right)\right|^2,$$
where we used the conventions of Theorem~\ref{t:evolution}.
\end{prop}
This follows from Propostion~\ref{coherent2} combined with paragraph~\ref{two_coherent} and Theorem~\ref{t:evolution}. It shows that, in this particular case, the Quantum Loschmidt Echo is periodic in time for times scales of order $\tau_{\hbar}^c$.

\section{Semiclassical fidelity distributions}\label{s:fidelity-distrib}

Even if we are primarly interested in the study of the Quantum Loschmidt Echo, we will in fact study some slightly more general quantities which may 
be of independent interest and which already appeared in~\cite{MaRi16}. We shall call these intermediary objects \emph{semiclassical fidelity distributions}. Before defining them and 
relating them to the quantities of the introduction, let us first fix some conventions. Associated with the Schr\"odinger 
equations~\eqref{e:schrodinger-unperturbed} and~\eqref{e:schrodinger-perturbed} are two semiclassical operators acting on $L^2(\IT^2)$:
$$\widehat{P}_0(\hbar):=-\frac{\hbar^2\Delta}{2},\ \text{and}\ \widehat{P}_{\eps}(\hbar):=-\frac{\hbar^2\Delta}{2}+\eps_{\hbar}V.$$
We will always assume that assumption~\eqref{e:perturbation-size} on the size of the perturbation is satisfied. 
In order to define these semiclassical fidelity distributions, we fix two sequences of \emph{normalized} initial data $(\psi_{\hbar}^1)_{0<\hbar\leq 1}$ and 
$(\psi_{\hbar}^2)_{0<\hbar\leq 1}$ satisfying the frequency assumptions~\eqref{e:shosc} and~\eqref{e:hosc}. We then define the following 
\emph{semiclassical fidelity distribution} on $T^*\IT^2$:  
$$\forall t\in\IR,\ F_{\hbar}(t): a\in\ml{C}^{\infty}_c(T^*\IT^2)
\mapsto \left\la \psi_{\hbar}^1,e^{\frac{it\widehat{P}_{\eps}(\hbar)}{\hbar}}\Oph(a)e^{-\frac{it\widehat{P}_{0}(\hbar)}{\hbar}}\psi_{\hbar}^2\right\ra_{L^2(\IT^2)},$$
where $\Oph(a)$ is the standard quantization defined in appendix~\ref{a:sc-an}. The goal of this section is to describe the properties of their 
accumulation points. The main observation from this section is that the description of the 
Quantum Loschmidt Echo at the critical time scale $\tau_{\hbar}^c$ follows from the detailed analysis of these 
accumulation points -- see paragraph~\ref{ss:main-obs} and Proposition~\ref{p:decomp-fidelity}.

\subsection{Extracting subsequences}

Recall that we denote by $\tau_{\hbar}^c$ the critical time scale $\frac{\hbar}{\eps_{\hbar}}$. We first extract converging subsequences from these sequences of distributions. 
Let $a(t,x,\xi)$ be an element in $\ml{C}^{\infty}_c(\IR\times T^*\IT^2)$. 
 From the Calder\'on-Vaillancourt Theorem~\ref{t:cald-vail}, one knows that
\begin{equation}\label{e:cald-vail}\left|\int_{\IR} \la F_{\hbar}(t\tau_{\hbar}^c) , a(t)\ra dt \right|\leq C
\sum_{|\alpha|\leq D}\hbar^{\frac{|\alpha|}{2}}\int_{\IR}\|\partial_{x,\xi}^{\alpha}a(t)\|_{\infty}dt,\end{equation}
for some universal positive constants $C$ and $D$. In particular, this defines a sequence of bounded distributions on $\IR\times T^*\IT^2$. 
Thus, up to an extraction, there exists $F(t,x,\xi)$ in $\ml{D}'(\IR\times T^*\IT^2)$ such that, for every $a$ in $\ml{C}^{\infty}_c(\IR\times T^*\IT^2)$, 
one has
$$\lim_{\hbar_n\rightarrow 0^+}\int_{\IR} \la F_{\hbar_n}(t\tau_{\hbar_n}^c) , a(t)\ra dt  =\int_{T^*\IT^2\times \IR}a(x,\xi,t) F(dt,dx,d\xi).$$
\begin{rema}
 In order to alleviate the notations, we shall write $\hbar\rightarrow 0^+$ instead of $\hbar_n\rightarrow 0^+$ which is a 
 standard convention in semiclassical analysis.
\end{rema}
From~\eqref{e:cald-vail}, one knows that
$$\left|\int_{T^*\IT^2\times \IR}a(x,\xi,t) F(dt,dx,d\xi)\right|\leq C\int_{\IR}\|a(t)\|_{\ml{C}^0}dt.$$
Thus, for a.e. $t$ in $\IR$, $F(t)$ defines an element of the Banach space\footnote{More generally, 
for a locally compact metric space $X$, we denote by $\ml{M}(X)$ the set of finite complex Radon measure on $X$.} 
$\ml{M}(T^*\IT^2)$ of finite (complex) Radon measures on $T^*\IT^2$. By an approximation argument, 
one can also verify that, for every $\theta$ in $L^1(\IR)$ and for every $a$ in $\ml{C}^{\infty}_c(T^*\IT^2)$, 
one has
$$\lim_{\hbar\rightarrow 0^+}\int_{\IR} \theta(t)\la F_{\hbar}(t\tau_{\hbar}^c) , a\ra dt  =\int_{\IR}\theta(t)\left(\int_{T^*\IT^2}a(x,\xi) F(t,dx,d\xi)\right)dt.$$
Note that compared with the classical case of semiclassical measures~\cite{Ge91, Ma09, Zw12}, $F(t)$ is a priori only a \emph{complex measure}. We also remark 
that, thanks to the frequency assumption~\eqref{e:shosc}, one has, for a.e. $t$ in $\IR$,
\begin{equation}\label{e:no-mass-zero-section}
 |F(t)|(\IT^2\times\{0\})=0.
\end{equation}
Up to another extraction, we can also suppose that there exists some finite (complex) Radon measure $F_0$ on $T^*\IT^2$ such that, for every $a$ 
in $\ml{C}^{\infty}_c(T^*\IT^2)$,
$$\lim_{\hbar\rightarrow 0^+}\la F_{\hbar}(0) , a\ra  =\int_{T^*\IT^2}a(x,\xi) F_0(dx,d\xi).$$

\textbf{From this point on, we fix the accumulation point $F(t)$ and we want to describe it in terms of $t$}. Let us start with the following lemma:
\begin{lemm}[Invariance by the geodesic flow]\label{l:invariance-geodesic-flow} Denote by $\varphi^s$ the geodesic flow, i.e.
$$\forall (x,\xi)\in T^*\IT^2,\ \forall s\in\IR,\ \varphi^s(x,\xi):=(x+s\xi,\xi).$$
 Then, for every $a$ in $\ml{C}^{\infty}_c(T^*\IT^2)$ and for a.e. $t$ in $\IR$, one has
 $$\forall s\in\IR,\ \int_{T^*\IT^2}a\circ\varphi^s(x,\xi)F(t,dx,d\xi)=\int_{T^*\IT^2}a(x,\xi)F(t,dx,d\xi).$$
\end{lemm}
This Lemma shows that, as for the case of semiclassical measures~\cite{Ma09}, the limit object we obtain is invariant by the geodesic flow. We 
emphasize that it is important here to have $\eps_{\hbar}\ll\hbar$.
\begin{proof}
 As for the extraction argument, the proof of this Lemma is the same as for semiclassical measure. Yet, let us recall the proof as it is instructive regarding the 
 upcoming proofs. We write
\begin{eqnarray*}\frac{d}{dt}\la F_{\hbar}(t\tau_{\hbar}^c) , a\ra & = & 
\frac{i\tau_{\hbar}^c}{\hbar}\left\la   \psi_{\hbar}^1,e^{\frac{it\tau_{\hbar}^c\widehat{P}_{\eps}(\hbar)}{\hbar}}\left[\widehat{P}_0(\hbar),\Oph(a)\right]e^{-\frac{it\tau_{\hbar}^c\widehat{P}_{0}(\hbar)}{\hbar}}\psi_{\hbar}^2\right\ra_{L^2(\IT^2)}\\
 & &+\frac{i\tau_{\hbar}^c\eps_{\hbar}}{\hbar}\left\la   \psi_{\hbar}^1,e^{\frac{it\tau_{\hbar}^c\widehat{P}_{\eps}(\hbar)}{\hbar}}V \Oph(a)e^{-\frac{it\tau_{\hbar}^c\widehat{P}_{0}(\hbar)}{\hbar}}\psi_{\hbar}^2\right\ra_{L^2(\IT^2)}
\end{eqnarray*}
As we use the standard quantization, we have that $V \Oph(a)=\Oph(Va)$ which is a bounded operator thanks to the Calder\'on-Vaillancourt 
Theorem~\ref{t:cald-vail}. Moreover, using the composition Theorem~\ref{t:composition} for pseudodifferential operators and the 
Calder\'on-Vaillancourt Theorem one more time, we have that $\left[\widehat{P}_0(\hbar),\Oph(a)\right]=\frac{\hbar}{i}\Oph(\xi.\partial_xa)
+\ml{O}_{L^2\rightarrow L^2}(\hbar^2).$ Using our assumption~\eqref{e:perturbation-size} on the size of $\eps_{\hbar}$, we find that
$$\frac{d}{dt}\la F_{\hbar}(t\tau_{\hbar}^c) , a\ra=i\tau_{\hbar}^c\la F_{\hbar}(t\tau_{\hbar}^c) , \xi.\partial_xa\ra+o(\tau_{\hbar}^c).$$
Integrating this relation against $\theta$ in $\ml{C}^{\infty}_c(\IR)$, we find that
$$\frac{i}{\tau_{\hbar}}\int_{\IR}\theta'(t)\la F_{\hbar}(t\tau_{\hbar}^c) , a\ra dt=\int_{\IR}\theta(t)\la F_{\hbar}(t\tau_{\hbar}^c) , \xi.\partial_xa\ra dt+o(1),$$
which concludes the proof by letting $\hbar\rightarrow 0^+$.
\end{proof}

Finally, we note that, up to another extraction, we can suppose that, for a.e. $t$ in $\IR$, there exists $\nu(t)\in \ml{M}(\IT^2)$ 
such that, for every $\theta$ in $L^1(\IR)$ and for every $a\in\ml{C}^{\infty}(\IT^2)$,
$$\lim_{\hbar\rightarrow 0^+}\int_{\IR} \theta(t)\la F_{\hbar}(t\tau_{\hbar}^c) , a\ra dt  =
\int_{\IR}\theta(t)\left(\int_{T^*\IT^2}a(x) \nu(t,dx)\right)dt.$$
Thanks to the frequency assumption~\eqref{e:hosc}, we know that there is no escape of mass at infinity. Therefore, one has
\begin{equation}\label{e:pushforward}
 \nu(t,x)=\int_{\IR^2}F(t,x,d\xi).
\end{equation}

\subsection{Time evolution}\label{ss:main-obs} Let us now discuss the relation of these fidelity distributions with the quantities appearing in the introduction. There, 
we were mostly interested in the case where $a=1$. In that particular case, we have that
$$\frac{d}{dt}\left(e^{-it\int_{\IT^2} V}\la F_{\hbar}(t\tau_{\hbar}^c),1\ra\right)=ie^{-it\int_{\IT^2} V}\left\la 
F_{\hbar}(t\tau_{\hbar}^c),\left(V-\int_{\IT^2}V\right)\right\ra,$$
or equivalently, for every $t\in\IR$ 
$$\la F_{\hbar}(t\tau_{\hbar}^c),1\ra=e^{it\int_{\IT^2} V}\la F_{\hbar}(0),1\ra+i\int_{0}^te^{i(t-t')\int_{\IT^2} V}\left\la 
F_{\hbar}(t'\tau_{\hbar}^c),\left(V-\int_{\IT^2}V\right)\right\ra dt'.$$
Letting $\hbar\rightarrow 0^+$, we find that, for every $t\in\IR$,
\begin{equation}\label{e:key-formula-loschmidt}
 \lim_{\hbar\rightarrow 0^+} \la F_{\hbar}(t\tau_{\hbar}^c),1\ra=e^{it\int_{\IT^2} V}\la F_0,1\ra
 +i\int_{0}^te^{i(t-t')\int_{\IT^2} V}\left(\int_{T^*\IT^2} \left(V(x)-\int_{\IT^2}V\right)F(t',dx,d\xi)\right) dt',
\end{equation}
or equivalently
$$\lim_{\hbar\rightarrow 0^+} \la F_{\hbar}(t\tau_{\hbar}^c),1\ra=e^{it\int_{\IT^2} V}\la F_0,1\ra
 +i\int_{0}^te^{i(t-t')\int_{\IT^2} V}\left(\int_{\IT^2} \left(V(x)-\int_{\IT^2}V\right)\nu(t',dx)\right) dt'.$$
Note that, in the particular case where $\psi_{\hbar}^1=\psi_{\hbar}^2=\psi_{\hbar}$, the first term of the right hand side is equal to 
$1$. Hence, \textbf{describing the Quantum Loschmidt echo at the critical time scale boils down to the description of the fidelity 
distribution} $F(t)$ (more precisely of its pushforward $\nu(t)$ on $\IT^2$) in terms of $t$ and of the initial data.

\begin{rema}
 Note that, up to this point, our analysis did not really used the fact that we are on the torus and it could be adapted to deal with 
 more general Riemannian manifolds.
\end{rema}

\subsection{Decomposition of phase space}

In order to describe $F(t)$ in terms of $t$, we will first exploit its invariance under the geodesic flow in order to decompose it into infinitely 
many pieces indexed by the family $\ml{L}$ of primitive sublattices of $\IZ^2$. We follow here the presentation of~\cite{AM10}. Recall that 
a sublattice $\Lambda$ is said to be primitive if $\la\Lambda\ra\cap\IZ^2=\Lambda$, where $\la\Lambda\ra$ is the subspace of $\IR^2$ 
spanned by $\Lambda$. For $\Lambda$ in $\ml{L}$, we introduce 
$$\Lambda^{\perp}:=\left\{\xi\in\IR^2:\xi.k=0,\ \forall k\in\Lambda\right\}.$$
To every fixed covector $\xi\in\IR^2$, we also associate the sublattice
$$\Lambda_{\xi}:=\left\{k\in\IZ^2:k.\xi=0\right\},$$
and, for every $0\leq j\leq 2$, we denote by $\Omega_j\subset\IR^2$ the following subsets of covectors
$$\Omega_{0}:=\{0\},\ \Omega_1:=\left\{\xi\in\IR^2:\text{rk}\Lambda_{\xi}=1\right\},\ \text{and}\ \Omega_2=\IR^2-(\Omega_0\cup\Omega_1).$$
The fact that $\xi\in\Omega_j$ is equivalent to say that the orbit $\{\varphi^s(x,\xi)\}$ fills a torus of dimension $j$. We also define
$$R_{\Lambda}:=\Lambda^{\perp}\cap\Omega_{2-\text{rk}(\Lambda)}.$$
Note that, for $\Lambda$ of rank $1$, we have $R_{\Lambda}=\Lambda^{\perp}-\{0\}$, and that we have the following partition of $\IR^2$ indexed by the primitive sublattices of 
$\IZ^2$:
$$\IR^2=\bigsqcup_{\Lambda\in\ml{L}}R_{\Lambda}.$$

Let us now decompose $F(t)$ according to this partition of $T^*\IT^2$. In fact, the above discussion allows to write 
two natural decompositions of $F(t)$:
\begin{equation}\label{e:decomp-mu}
F(t)=\sum_{\Lambda\in\ml{L}}F(t)\rceil_{\IT^2\times R_{\Lambda}},
\end{equation}
and its Fourier decomposition
$$F(t,x,\xi)=\sum_{k\in\IZ^2}\widehat{F}_k(t,\xi)e^{2i\pi k.x}.$$
For $\Lambda\in\ml{L}$, we denote by $\ml{I}_{\Lambda}(F(t))$ the distribution
$$\ml{I}_{\Lambda}(F(t)):=\sum_{k\in\Lambda}\widehat{F}_k(t,\xi)e^{2i\pi k.x}.$$
Note that this is consistent with the conventions we have introduced in section~\ref{s:main-result}. 
The following result holds (see section $2$ in~\cite{AM10}):
\begin{prop}\label{p:decomp-meas} For a.e. $t\in\IR$, one has
\begin{enumerate}
 \item for every $\Lambda\in\ml{L}$, the distribution $\ml{I}_{\Lambda}(F(t))$ is a finite complex Radon measure on $T^*\IT^2$; 
 \item every term in~\eqref{e:decomp-mu} is a finite complex Radon measure invariant by $\varphi^s$ and
\begin{equation}\label{e:inv-orth}F(t)\rceil_{\IT^2\times R_{\Lambda}}=\ml{I}_{\Lambda}(F(t))\rceil_{\IT^2\times R_{\Lambda}}.\end{equation}
\end{enumerate}
Finally, property~\eqref{e:inv-orth} is equivalent to the fact that $F(t)\rceil_{\IT^2\times R_{\Lambda}}$ is invariant by the translations:
$$\tau^{v}:(x,\xi)\mapsto (x+v,\xi),\ \text{for every}\ v\in\Lambda^{\perp}.$$
\end{prop}

\begin{rema}
 The proof in~\cite{AM10} was given for finite positive measure but it can be adapted verbatim to fit our framework where we have to deal with 
 finite complex Radon measures.
\end{rema}

To summarize, we can decompose the distribution we are interested in as follows:
\begin{equation}\label{e:decompose-fidelity}
 F(t)=\ml{I}_{0}(F(t))\rceil_{\IT^2\times R_{\Lambda}}+\sum_{\Lambda:\text{rk}(\Lambda)=1}\ml{I}_{\Lambda}(F(t))\rceil_{\IT^2\times \Lambda^{\perp}-\{0\}}.
\end{equation}
Note that we do not have any term associated with $\Lambda=\IZ^2$ thanks to the frequency assumption~\eqref{e:shosc} -- see~\eqref{e:no-mass-zero-section}. Recall 
from~\eqref{e:key-formula-loschmidt} that we are interested in determining $\la F(t),V-\int_{\IT^2}V\ra$ or more precisely
$$\int_0^te^{-it'\int_{\IT^2 V}}\left\la F(t'), V-\int_{\IT^2} V\right\ra dt' = 
\sum_{\Lambda:\text{rk}(\Lambda)=1} \int_0^te^{-it'\int_{\IT^2 }V}\left\la
F(t'),\ml{I}_{\Lambda}(V)-\int_{\IT^2} V\right\ra dt'.$$
Therefore, applying~\eqref{e:decompose-fidelity}, we have shown the following which is \textbf{the main observation of this section}:
\begin{prop}\label{p:decomp-fidelity} Using the above conventions, we have
\begin{eqnarray*}\lim_{\hbar\rightarrow 0^+} \la F_{\hbar}(t\tau_{\hbar}^c),1\ra & = & e^{it\int_{\IT^2} V}\la F_0,1\ra\\
  & + &i\sum_{\Lambda:\operatorname{rk}(\Lambda)=1}
  \int_{0}^te^{i(t-t')\int_{\IT^2} V}\left\la
\ml{I}_{\Lambda}(F(t'))\rceil_{\IT^2\times \Lambda^{\perp}-\{0\}},\ml{I}_{\Lambda}(V)-\int_{\IT^2} V\right\ra dt'.
 \end{eqnarray*}
\end{prop}

Hence, this formula combined with~\eqref{e:key-formula-loschmidt} allows to 
reduce the problem to analyzing the fidelity distribution along the submanifolds $\IT^2\times \Lambda^{\perp}$ (for every rank $1$ sublattice). 
In order to this, we will proceed to a \emph{second microlocalization along these submanifolds} following the strategies 
from~\cite{Ma10, AM10, AFM12}.

\section{Set-up of the two-microlocal tools}\label{s:twomicrolocal}

We will introduce two-microlocal objects in order to proceed to the analysis of 
the fidelity distribution near the submanifolds $\IT^2\times\Lambda^{\perp}\subset T^*\IT^2$. 
For that purpose, we will make use of the tools developped in~\cite{Ma10, AM10, AFM12} that we 
will briefly review in this section using the conventions from~\cite{MaRi16b}. We note that the main differences with these references are the choice of rescaling 
for the second microlocalization, and the nature of the propagation relation that comes out of the analysis. For instance, the 
potential appears in the Hamiltonian dynamics induced along $\Lambda$ in reference~\cite{MaRi16b} while here it will play a role as a phase factor -- see 
Proposition~\ref{p:propag-compact-part}.

To proceed with our analysis, we fix $\Lambda$ a primitive sublattice of rank $1$ and we denote by 
$\widehat{\IR}$ the compactified space $\IR\cup\{\pm\infty\}$. Then, we introduce an auxiliary distribution, for every 
$a\in\ml{C}^{\infty}_c(T^*\IT^2\times\widehat{\IR})$,
$$\la F_{\Lambda,\hbar}(t\tau_{\hbar}^c),a\ra:=\left\la \psi_{\hbar}^1,e^{\frac{it\tau_{\hbar}^c\widehat{P}_{\eps}(\hbar)}{\hbar}}
\Oph\left(a\left(x,\xi,\frac{\hbar H_{\Lambda}(\xi)}{\eps_{\hbar}}\right)\right)
e^{-\frac{it\tau_{\hbar}^c\widehat{P}_{0}(\hbar)}{\hbar}}\psi_{\hbar}^2\right\ra_{L^2(\IT^2)}.$$
Note that these quantities are slightly more general than the ones introduced in~\eqref{e:2micro-initialdata} as they depend on $(t,\xi)$ and 
as they extend to $\pm\infty$. We shall compare with definition~\eqref{e:2micro-initialdata} in paragraph~\ref{ss:pushforward}.

\begin{rema} Recall that semiclassical measures involving a second microlocalization primarly appeared in~\cite{Mi96, Ni96, Fe00, FeGe02} outside the context 
of integrable systems discussed in the above references. 
\end{rema}

The purpose of this section is to study the properties of these distributions and their relation to the 
fidelity distributions we have already defined. We proceed in several stages. First, we recall how one can extract 
converging subsequences and we explain how to decompose the limit distribution into 
two components (paragraph~\ref{ss:extract-2-microlocal}): the ``compact'' component and the one at infinity. After that, we show some invariance 
properties of the limiting distributions and relate these quantities to the ones we are primarly interested in (paragraph~\ref{ss:invariance}).

\begin{rema}\label{r:useful} We point out the following useful observations. First, if $a$ is an element in $\ml{C}^{\infty}_c(T^*\IT^2)$, we recover the fidelity distribution we have introduced 
before. Hence, this new distribution should be understood as a generalization of $F_{\hbar}(t\tau_{\hbar}^c)$ which captures 
some informations on the distribution near $\IT^2\times\Lambda^{\perp}$. We 
also note that this quantity is well defined as the standard quantization allows us to consider any observable which has bounded derivatives in the 
$x$ variables -- see appendix~\ref{a:sc-an}. Finally, we emphasize that
$$\Oph\left(a\left(x,\xi,\frac{\hbar H_{\Lambda}(\xi)}{\eps_{\hbar}}\right)\right)=\Op_{\hbar^2\eps_{\hbar}^{-1}}
\left(a\left(x,\frac{\eps_{\hbar}\xi}{\hbar}, H_{\Lambda}(\xi)\right)\right),$$
where we remind that $\hbar^2\ll\eps_{\hbar}\ll\hbar$.
\end{rema}

\subsection{Extracting converging subsequences}\label{ss:extract-2-microlocal}

As for the semiclassical fidelity distributions, we would like to extract subsequences $\hbar_{n}\rightarrow 0$ 
such that $F_{\Lambda,\hbar_n}(t\tau_{\hbar_n}^c)$ converges 
in a certain weak sense. For that purpose, we shall follow the more or 
less standard procedures of~\cite{Ge91, Ma09, Zw12} in the case of semiclassical measures. We denote by
$$\mathcal{B}:=\mathcal{C}^0_0(\IR^2\times\widehat{\IR},\mathcal{C}^3(\IT^2)),$$
the space of continuous function on $\IR^2\times\widehat{\IR}$ with values in $\mathcal{C}^3(\IT^2)$ and which tends to $0$ at 
infinity\footnote{Here, infinity means in the variables corresponding to $\IR^2$ as $\widehat{\IR}$ is compact.}. 
We endow this space with its natural topology of Banach space. According to the Calder\'on-Vaillancourt Theorem from the appendix, one knows that, 
for every $a\in L^1(\IR,\ml{B})$, one has
\begin{equation}\label{e:measure}\int_{\IR}|\la F_{\Lambda,\hbar}(t\tau_{\hbar}),a(t)\ra|dt\leq C\int_{\IR}\|a(t)\|_{\mathcal{B}}dt.\end{equation}
In other words, the map $t\mapsto F_{\Lambda,\hbar}(t\tau_{\hbar})$ defines a bounded sequence in $L^1(\IR,\ml{B})'$ endowed with its weak-$\star$ topology. 
Hence, after 
extracting a subsequence, one finds that there exists for a.e. $t$ in $\IR$ some element $F_{\Lambda}(t)$ in $\ml{B}'$ such that, for 
every $a$ in $\ml{C}^{\infty}_c(\IR\times T^*\IT^2\times\widehat{\IR})$, one has
$$\lim_{\hbar\rightarrow 0^+}\int_{\IR\times T^*\IT^2\times\widehat{\IR}}a(t,x,\xi,\eta)F_{\Lambda,\hbar}(t\tau_{\hbar},dx,d\xi,d\eta)dt
=\int_{\IR}\left(\int_{T^*\IT^2\times\widehat{\IR}}a(t,x,\xi,\eta) F_{\Lambda}(t,dx,d\xi,d\eta)\right)dt.$$
\begin{rema}
 Here, we used the first part of the Calder\'on-Vaillancourt. If we had used the second part, it would just have changed slightly the Banach space 
 involved in our argument.
\end{rema}

By a density argument, one can find that, for every $\theta$ in $L^1(\IR)$ and for every $a$ in $\ml{C}^{\infty}_c(T^*\IT^2\times\widehat{\IR})$,
$$\lim_{\hbar\rightarrow 0^+}\int_{\IR}\theta(t)\la F_{\Lambda,\hbar}(t\tau_{\hbar}),a\ra dt =\int_{\IR}\theta(t)\la F_{\Lambda}(t),a\ra dt.$$
These limiting functionals are related to $F(t)$ in the following manner:
\begin{equation}\label{e:pushforward-2microlocal}F(t)=\int_{\widehat{\IR}}F_{\Lambda}(t,d\eta).\end{equation}
The main objective of the rest of the article is to describe the properties of $F_{\Lambda}(t)$ in function of $t$ and of the initial data. 
Regarding this aim, it is convenient to split these two-microlocal distributions in two parts: the ``compact'' one and the one 
at ``infinity'' (in the $\widehat{\IR}$ variable). Before doing that, we observe that, up to another extraction, we can also suppose that 
there exists $F_{\Lambda}^0$ in $\ml{B}'$ such that, for every $a$ in $\ml{C}^{\infty}_c(T^*\IT^2\times\widehat{\IR})$,
$$\lim_{\hbar\rightarrow 0^+}\la F_{\Lambda,\hbar}(0),a\ra =\la F_{\Lambda}^0,a\ra.$$
Therefore, we would like to relate $F_{\Lambda}(t)$ to $F_{\Lambda}^0$. Finally, up to some diagonal 
extraction argument, we can suppose that these differerent linear functionals converge for any primitive sublattice 
$\Lambda$ of rank $1$ along the same subsequence.

At this point, $F_{\Lambda}(t)$ is \emph{only an element in the dual of $\ml{B}$ and 
not a priori a complex Radon measure}. Yet, this can be overcome via the application of the second part of the 
Calder\'on-Vaillancourt Theorem~\ref{t:cald-vail}. In fact, thanks to this Theorem, one knows that, for every $\theta$ in $L^1(\IR)$ and for every 
$a\in\ml{C}^{\infty}_c(T^*\IT^2\times\widehat{\IR})$, one has
$$\left|\int_{\IR}\theta(t)\la F_{\Lambda,\hbar}(t\tau_{\hbar}),a\ra dt\right|\leq C\int_{\IR}|\theta(t)|\|a\|_{\mathcal{C}^0}dt+\ml{O}(\hbar^2\eps_{\hbar}^{-1}).$$
Hence, passing to the limit $\hbar\rightarrow 0^+$, this gives
$$\left|\int_{\IR}\theta(t)\la F_{\Lambda}(t),a\ra dt\right|\leq C\int_{\IR}|\theta(t)|\|a\|_{\mathcal{C}^0}dt.$$
We emphasize that we crucially use here that $\hbar^2\ll\eps_{\hbar}$, which yields comfortable simplifications compared with~\cite{AM10, AFM12} where $\eps_{\hbar}=\hbar^2$ (and thus the limit objects are 
not a priori measures). In particular, we find that, 
for a.e. $t$ in $\IR$, $F_{\Lambda}(t)$ belongs to $\ml{M}(T^*\IT^2\times\widehat{\IR})$. Thus, we can split $F_{\Lambda}(t)$ as
\begin{equation}\label{e:split-distrib}
 F_{\Lambda}(t)=\tilde{F}_{\Lambda}(t)+\tilde{F}^{\Lambda}(t),
\end{equation}
where
$$\tilde{F}_{\Lambda}(t):=F_{\Lambda}(t)\rceil_{\IR^2\times\IR}
,\text{ and }\tilde{F}^{\Lambda}(t):=F_{\Lambda}(t)\rceil_{\IR^2\times\{\pm\infty\}}.$$

Finally, we can split similarly $F_{\Lambda}^0$ in two parts that we denote by $\tilde{F}_{\Lambda}^0$ and $\tilde{F}^{\Lambda,0}.$

\subsection{First invariance properties}\label{ss:invariance}

Recall from Proposition~\ref{p:decomp-fidelity} that we aim at describing $$\ml{I}_{\Lambda}(F(t))\rceil_{\IT^2\times\Lambda^{\perp}-\{0\}}.$$ 
Thanks to~\eqref{e:pushforward-2microlocal}, this is related to the 
two-microlocal quantities we have introduced as follows:
\begin{equation}\label{e:split-F}
 \ml{I}_{\Lambda}(F(t))\rceil_{\IT^2\times\Lambda^{\perp}-\{0\}}=\int_{\IR}\ml{I}_{\Lambda}(\tilde{F}_{\Lambda})(t,d\eta)\rceil_{\IT^2\times\Lambda^{\perp}-\{0\}}+
 \int_{\{\pm\infty\}}\ml{I}_{\Lambda}(\tilde{F}^{\Lambda})(t,d\eta)\rceil_{\IT^2\times\Lambda^{\perp}-\{0\}}.
\end{equation}
As a first step, we will show the following result concerning the part at infinity which is the analogue in our context of~\cite[Th.~13(ii)]{AM10}:
\begin{lemm} Let $\Lambda$ be rank $1$ primitive sublattice.
Then, for every $k$ in $\Lambda-\{0\}$, for every $a$ in $\ml{C}^{\infty}_c(\IR^2\times\widehat{\IR})$ and for a.e. $t$ in $\IR$,
$$\la \tilde{F}^{\Lambda}(t),a(\xi,\eta)e^{-2i\pi k.x}\ra=0.$$
\end{lemm}
Before proving this Lemma, observe that it allows us to rewrite the quantities appearing in Proposition~\ref{p:decomp-fidelity} as
\begin{equation}\label{e:decompose-key-formula-2}
\left\la\ml{I}_{\Lambda}(F(t'))\rceil_{\IT^2\times \Lambda^{\perp}-\{0\}},\ml{I}_{\Lambda}(V)-\int_{\IT^2} V\right\ra 
=\left\la\int_{\IR}\ml{I}_{\Lambda}(\tilde{F}_{\Lambda}(t',d\eta))\rceil_{\IT^2\times \Lambda^{\perp}-\{0\}},\ml{I}_{\Lambda}(V)-\int_{\IT^2} V\right\ra,
\end{equation}
 for a.e. $t$ in $\IR$. Recall from Proposition~\ref{p:decomp-fidelity} that this is exactly what remains to be computed if we want 
to find an expression for the Quantum Loschmidt Echo at the critical time scale $\tau_{\hbar}^c$. Hence, our analysis boils down to the description of the 
``compact'' part $\tilde{F}_{\Lambda}(t)$ of $F_{\Lambda}(t)$.

\begin{proof} Let $\chi$ be a smooth cutoff function on $\IR$ which is equal to $1$ near $0$ and to $0$ outside a small neighborhood of $0$. We fix 
$a$ in $\ml{C}^{\infty}_c(\IR^2\times\widehat{\IR})$ and $R>0$. Then, we define
$$a^R(\xi,\eta)=\left(1-\chi\left(\frac{\eta}{R}\right)\right)a(\xi,\eta).$$
Let $\Lambda$ be a primitive rank one sublattice and $k\in\Lambda-\{0\}.$ Recall that we use the standard 
quantization (see appendix~\ref{a:sc-an} for a brief reminder). Hence, we have the identity
\begin{equation}\label{deriv}
\begin{split}
&\frac{d}{dt}\left\la F_{\Lambda,\hbar}(t\tau_{\hbar}^c),\frac{a^R(\xi,\eta) e^{-2i\pi k.x}}{\eta}\right\ra\\
=&\frac{i\tau_\hbar^c}{\hbar}\left\la \psi_\hbar^1,e^{\frac{it\tau_{\hbar}^c\widehat{P}_{\eps}(\hbar)}{\hbar}}\left[-\frac{\hbar^2\Delta}{2},
\Op_\hbar\left(\frac{\eps_{\hbar}}{\hbar H_{\Lambda}(\xi)}
 a^R\left(\xi,\frac{\hbar H_{\Lambda}(\xi)}{\eps_{\hbar}}\right)e^{-2i\pi k.x}\right)\right]e^{-\frac{it\tau_{\hbar}^c\widehat{P}_{0}(\hbar)}{\hbar}}\psi_\hbar^2\right\ra\\
&+i\left\la \psi_\hbar^1,
e^{\frac{it\tau_{\hbar}^c\widehat{P}_{\eps}(\hbar)}{\hbar}}\Op_\hbar\left(\frac{\eps_{\hbar}}{\hbar H_{\Lambda}(\xi)}
 a^R\left(\xi,\frac{\hbar H_{\Lambda}(\xi)}{\eps_{\hbar}}\right)e^{-2i\pi k.x}V(x)\right)e^{-\frac{it\tau_{\hbar}^c\widehat{P}_{0}(\hbar)}{\hbar}}\psi_\hbar^2\right\ra.
\end{split}
\end{equation}
For the first term on the right hand side, we have, using the composition 
formula for the standard quantization (see appendix~\ref{a:sc-an}), that 
 $$\left[-\frac{\hbar^2\Delta}{2},\Op_\hbar\left(\frac{\eps_{\hbar}}{\hbar H_{\Lambda}(\xi)}
 a^R\left(\xi,\frac{\hbar H_{\Lambda}(\xi)}{\eps_{\hbar}}\right)e^{-2i\pi k.x}\right)\right]\hspace{4cm}$$
$$ =\Op_\hbar\left(a^R\left(\xi,\frac{\hbar H_{\Lambda}(\xi)}{\eps_{\hbar}}\right)e^{-2i\pi k.x}
\left(-2\eps_{\hbar}\pi H_{\Lambda}(k)-\frac{2\pi^2\hbar^2\eps_{\hbar}H_{\Lambda}(k)^2}{\hbar H_{\Lambda}(\xi)} \right)\right),$$
where we used the fact that $k$ belongs to $\Lambda$. Thanks to the Calder\'on-Vaillancourt Theorem~\ref{t:cald-vail}, we can deduce that
$$\left[-\frac{\hbar^2\Delta}{2},\Op_\hbar\left(\frac{\eps_{\hbar}}{\hbar H_{\Lambda}(\xi)}
 a^R\left(\xi,\frac{\hbar H_{\Lambda}(\xi)}{\eps_{\hbar}}\right)e^{-2i\pi k.x}\right)\right]\hspace{4cm}$$
 $$=-2\eps_{\hbar}\pi H_{\Lambda}(k)\Op_\hbar\left(a^R\left(\xi,\frac{\hbar H_{\Lambda}(\xi)}{\eps_{\hbar}}\right)e^{-2i\pi k.x}
\right)+\ml{O}_{L^2\rightarrow L^2}(\hbar^2 R^{-1}).$$
Similarly, we find using the Calder\'on-Vaillancourt Theorem that
$$\Op_\hbar\left(\frac{\eps_{\hbar}}{\hbar H_{\Lambda}(\xi)}
 a^R\left(\xi,\frac{\hbar H_{\Lambda}(\xi)}{\eps_{\hbar}}\right)e^{-2i\pi k.x}V(x)\right)=\ml{O}_{L^2\rightarrow L^2}(R^{-1}).$$
Implementing these two equalities in~\eqref{deriv}, we find that
\begin{equation}\label{deriv2}
\begin{split}
&\frac{d}{dt}\left\la F_{\hbar}(t\tau_{\hbar}^c),\frac{a^R(\xi,\eta)e^{-2i\pi k.x}}{\eta}\right\ra\\
=&-2i\pi H_{\Lambda}(k)\left\la F_{\Lambda,\hbar}(t\tau_{\hbar}),a^R(\xi,\eta) e^{-2i\pi k.x}\right\ra
+\ml{O}\left(R^{-1}\left(\tau_{\hbar}^c\hbar+1\right)\right).
\end{split}
\end{equation}
Recall that $\tau_{\hbar}^c=\frac{\hbar}{\eps_{\hbar}}\leq\frac{1}{\hbar}$. Hence, the remainder in this equality is of order $\ml{O}(R^{-1})$.
Another application of the Calder\'on-Vaillancourt Theorem yields
$$\left\la F_{\Lambda,\hbar}(t\tau_{\hbar}^c),\frac{a^R(\xi,\eta)e^{-2i\pi k.x}}{\eta}\right\ra=\ml{O}(R^{-1}).$$
Thus, if we fix $\theta$ in $\ml{C}^{\infty}_c(\IR)$, we find after integrating by parts in~\eqref{deriv2} that
$$\int_{\IR}\theta(t)\left\la F_{\Lambda,\hbar}(t\tau_{\hbar}^c),a^R(\xi,\eta)e^{-2i\pi k.x}\right\ra dt=\ml{O}(R^{-1}).$$
Finally, if we let $\hbar$ go to $0$ and $R$ to $+\infty$ (in this order), 
we find
$$\int_{\IR}\theta(t)\left\la F_{\Lambda}(t)\rceil_{T^*\IT^2\times\{\pm\infty\}},a(\xi,\eta)e^{-2i\pi k.x}\right\ra dt=0.$$
This is valid for any $\theta$ in $\ml{C}^{\infty}_c(\IR)$ and thus concludes the proof of the Lemma.
\end{proof}

We conclude this section by showing that $F_{\Lambda}(t)$ is also invariant under the geodesic flow
\begin{lemm}\label{l:invariance-geodesic-flow-two-microlocal} Let $\Lambda$ be a rank $1$ primitive sublattice. Then, for every $a$ in 
$\ml{C}^{\infty}_c(T^*\IT^2\times\widehat{\IR})$ and for a.e. $t$ in $\IR$, one has
$$\la F_{\Lambda}(t), \xi.\partial_x a\ra=0.$$
Equivalently, $F_{\Lambda}(t)$ is invariant by the geodesic flow on $T^*\IT^2$.
\end{lemm}

\begin{proof} We fix $\chi$ to be a smooth cutoff function on $\IR$ which is compactly supported and equal to $1$ near $0$. 
For every $R>0$, we define $\chi_R(\eta)=\chi(\eta/R)$. Let $a$ be an element in $\ml{C}^{\infty}_c(T^*\IT^2\times\widehat{\IR})$. 
We define $a_R=\chi_R a$ and $a^R=(1-\chi_R)a$. Recall that we use the standard quantization and 
we have the identity
\begin{equation}\label{deriv0}
\begin{split}
&\frac{d}{dt}\la F_{\Lambda,\hbar}(t\tau_{\hbar}^c),a_R\ra\\
=&\frac{i\tau_\hbar^c}{\hbar}\left\la \psi_\hbar^1,e^{\frac{it\tau_{\hbar}^c\widehat{P}_{\eps}(\hbar)}{\hbar}}\left[-\frac{\hbar^2\Delta}{2},
\Op_\hbar\left(a_R\left(x,\xi,\frac{\hbar H_{\Lambda}(\xi)}{\eps_{\hbar}}\right)\right)\right]e^{\frac{it\tau_{\hbar}^c\widehat{P}_{0}(\hbar)}{\hbar}}\psi_\hbar^2\right\ra\\
&+i\left\la \psi_\hbar^1,e^{\frac{it\tau_{\hbar}^c\widehat{P}_{\eps}(\hbar)}{\hbar}}
\Op_\hbar\left(a_R\left(x,\xi,\frac{\hbar H_{\Lambda}(\xi)}{\eps_{\hbar}}\right) V(x)
\right)e^{\frac{it\tau_{\hbar}^c\widehat{P}_{0}(\hbar)}{\hbar}}\psi_\hbar^2\right\ra.
\end{split}
\end{equation}

Also, observe from the composition formula that 
$$\left[-\frac{\hbar^2\Delta}{2},\Op_\hbar\left(a_R\left(x,\xi,\frac{\hbar H_{\Lambda}(\xi)}{\eps_{\hbar}}\right)\right)\right]
=\Op_\hbar\left(\left(\frac{\hbar}{i}\xi\cdot\partial_x -\frac{\hbar^2}{2}\Delta_x \right)a_R\left(x,\xi,\frac{\hbar H_{\Lambda}(\xi)}{\eps_{\hbar}}\right)\right).$$ 
Hence, applying the Calder\'on-Vaillancourt Theorem implies 
that the first term on the RHS dominates, as $\hbar\to0^+$. We find that
$$\frac{d}{dt}\la F_{\hbar}(t\tau_{\hbar}^c),a_R\ra=\tau_{\hbar}^c\la F_{\hbar}(t\tau_{\hbar}^c),\xi.\partial_xa_R\ra+\ml{O}(1)
=\tau_{\hbar}^c\la F_{\hbar}(t\tau_{\hbar}^c),(\xi.\partial_xa)\chi_R\ra+\ml{O}(1).$$
As the Calder\'on-Vaillancourt Theorem only involves the derivatives in $x$, we can apply the same argument with $a^R$ instead of $a_R$ and we find
$$\frac{d}{dt}\la F_{\hbar}(t\tau_{\hbar}^c),a^R\ra=\tau_{\hbar}^c\la F_{\hbar}(t\tau_{\hbar}^c),(\xi.\partial_xa)(1-\chi_R)\ra+\ml{O}(1).$$
Summing these two relations and integrating them against a test function $\theta$ in $\ml{C}_{c}^{\infty}(\IR)$, we find, after letting $\hbar$ tend to $0$,
 $$\int_{\IR}\theta(t)\la F(t), \xi.\partial_xa\ra dt=0.$$
 As this is valid for every $\theta$, we can conclude the proof of the lemma.
\end{proof}

\subsection{Relation with $\mathbf{F}_{\Lambda}^0$}\label{ss:pushforward}

In this section, we introduced a quantity $\tilde{F}_{\Lambda}^0$ which is slightly different from the one we considered in~\eqref{e:2micro-initialdata}. Yet, both quantities are related 
thanks to the frequency assumption~\eqref{e:hosc}. In fact, this hypothesis implies that, for $R>0$ and for $a$ in $\ml{C}^{\infty}_c(\IT^2\times\IR)$,
\begin{eqnarray*}\left\la\psi_{\hbar},\Oph\left(\ml{I}_{\Lambda}(a)\left(x,\frac{\hbar H_{\Lambda}(\xi)}{\eps_{\hbar}}\right)\right)\psi_{\hbar}\right\ra
&=&\left\la\psi_{\hbar},\Oph\left(\ml{I}_{\Lambda}(a)\left(x,\frac{\hbar H_{\Lambda}(\xi)}{\eps_{\hbar}}\right)\right)\chi(-\hbar^2\Delta/R)\psi_{\hbar}\right\ra \\
&+ &r(R,\hbar),\end{eqnarray*}
where $\chi$ is a smooth and compactly supported function equal to $1$ near $0$ and where  $\lim_{R\rightarrow+\infty}\limsup_{\hbar\rightarrow 0^+}r(R,\hbar)=0.$ Applying the composition rules for 
pseudodifferential operators and letting $\hbar\rightarrow 0$, we find that
$$\la\mathbf{F}_{\Lambda}^0,a\ra=\la\ml{I}_{\Lambda}(\tilde{F}_{\Lambda}^0),a(x,\eta)\chi(\|\xi\|^2/R)\ra+o(1),$$
as $R\rightarrow +\infty$. From the dominated convergence Theorem, we conclude that
\begin{equation}\label{e:pushforward-2micro}
 \mathbf{F}_{\Lambda}^0(x,\eta)=\int_{\IR^2}\ml{I}_{\Lambda}(\tilde{F}_{\Lambda}^0)(x,d\xi,\eta).
\end{equation}

\section{Proof of Theorem~\ref{t:evolution}}\label{s:proof}

Thanks to Proposition~\ref{p:decomp-fidelity} and to~\eqref{e:decompose-key-formula-2}, it now remains to 
determine $\ml{I}_{\Lambda}(\tilde{F}_{\Lambda}(t))$ in terms of $t$ and of $\tilde{F}_{\Lambda}^0$ in order to conclude the proof 
of Theorem~\ref{t:evolution}. This will be the main purpose of this section. After deriving the exact expression for 
$\ml{I}_{\Lambda}(\tilde{F}_{\Lambda}(t))$, we will explain how to prove Theorem~\ref{t:evolution} in paragraph~\ref{ss:conclusion}.

\subsection{Preliminary remarks}
Rather than $\tilde{F}_{\Lambda}(t)$, we are interested in the restriction of 
$\ml{I}_{\Lambda}(\tilde{F}_{\Lambda}(t))$ to $\IT^2\times\Lambda^{\perp}-\{0\}\times\IR$. Yet, 
this distinction is essentially irrelevant due to the following observation:
\begin{lemm} Let $\Lambda$ be a rank $1$ primitive sublattice. 
For every $a\in\ml{C}^{\infty}_c(T^*\IT^2\times\IR)$ whose support does not intersect $\IT^2\times\Lambda^{\perp}\times\IR$, and 
for a.e. $t$ in $\IR$, one has
 $$\la F_{\Lambda}(t), a\ra =0.$$
\end{lemm}
From this lemma, we deduce that
$$\tilde{F}_{\Lambda}(t)=\tilde{F}_{\Lambda}(t)\rceil_{\IT^2\times\Lambda^{\perp}\times\IR}=\tilde{F}_{\Lambda}(t)\rceil_{\IT^2\times\Lambda^{\perp}-\{0\}\times\IR},$$
where the second equality follows from the frequency assumptions~\eqref{e:shosc}. Recall that $\tilde{F}_{\Lambda}(t)$ is an invariant complex Radon 
measure. Hence, thanks to Proposition~\ref{p:decomp-meas} and to Lemma~\eqref{l:invariance-geodesic-flow-two-microlocal}, we can write
\begin{equation}\label{e:support-2micro}\tilde{F}_{\Lambda}(t)=\ml{I}_{\Lambda}(\tilde{F}_{\Lambda}(t))\rceil_{\IT^2\times\Lambda^{\perp}-\{0\}\times\IR}.\end{equation}
\begin{proof} Let $a$  be an element in $\ml{C}^{\infty}_c(T^*\IT^2\times\IR)$ whose support does not intersect $\IT^2\times\Lambda^{\perp}\times\IR$. 
It is sufficient to observe that
$$\la F_{\Lambda,\hbar}(t\tau_{\hbar}^c),a\ra=\left\la\psi_{\hbar}^1,e^{\frac{it\tau_{\hbar}^c\widehat{P}_{\eps}(\hbar)}{\hbar}}
\Op_\hbar\left(a\left(x,\xi,\frac{\hbar H_{\Lambda}(\xi)}{\eps_{\hbar}}\right)\right)
e^{\frac{it\tau_{\hbar}^c\widehat{P}_{\eps}(\hbar)}{\hbar}}\psi_{\hbar}^2\right\ra$$
is equal to $0$ for $\hbar$ small enough thanks to our support assumption on $a$.
\end{proof}

\subsection{Propagation formulas for $\tilde{F}_{\Lambda}(t)$}\label{invariance} 
We can now express $\tilde{F}_{\Lambda}(t)$ in  terms of $t$ and of the initial data:
\begin{prop}\label{p:propag-compact-part} Let $\Lambda$ be a rank one primitive sublattice. Then $t\mapsto \tilde{F}_{\Lambda}(t)$ is continuous and 
one has, for every $a$ in $\ml{C}^{\infty}_c(T^*\IT^2\times\IR)$
 $$\la \tilde{F}_{\Lambda}(t),\ml{I}_{\Lambda}(a)\ra=\left\la
\ml{I}_{\Lambda}(\tilde{F}_{\Lambda}^0)\left(x-t\frac{\eta \vec{v}_{\Lambda}}{L_{\Lambda}},\xi,\eta\right) 
 e^{i\int_0^t\ml{I}_{\Lambda}(V)\left(x-s\frac{\eta \vec{v}_{\Lambda}}{L_{\Lambda}}\right)ds}
,\ml{I}_{\Lambda}(a)\right\ra$$
\end{prop}

\begin{proof} Let $a(x,\xi,\eta)$ be an element in $\ml{C}^{\infty}_c(T^*\IT^2\times\IR)$. In our to derive our equation, we start by differentiating the 
map $t\mapsto \la F_{\Lambda,\hbar}(t\tau_{\hbar}^c),\ml{I}_{\Lambda}(a)\ra$. Recalling that we are using the standard quantization, we find that
\begin{equation}\label{deriv3}
\begin{split}
&\frac{d}{dt}\la F_{\Lambda,\hbar}(t\tau_{\hbar}^c),\ml{I}_{\Lambda}(a)\ra\\
=&\frac{i\tau_\hbar^c}{\hbar}\left\la \psi_\hbar^1,e^{\frac{it\tau_{\hbar}^c\widehat{P}_{\eps}(\hbar)}{\hbar}}\left[-\frac{\hbar^2\Delta}{2},
\Op_\hbar\left(\ml{I}_{\Lambda}(a)\left(x,\xi,\frac{\hbar}{\eps_{\hbar}H_{\Lambda}(\xi)}\right)\right)\right]e^{\frac{it\tau_{\hbar}^c\widehat{P}_{0}(\hbar)}{\hbar}}\psi_\hbar^2\right\ra\\
&+i\left\la \psi_\hbar^1,e^{\frac{it\tau_{\hbar}^c\widehat{P}_{\eps}(\hbar)}{\hbar}}
\Op_\hbar\left(\ml{I}_{\Lambda}(a)\left(x,\xi,\frac{\hbar}{\eps_{\hbar}H_{\Lambda}(\xi)}\right) V(x)
\right)e^{\frac{it\tau_{\hbar}^c\widehat{P}_{0}(\hbar)}{\hbar}}\psi_\hbar^2\right\ra.
\end{split}
\end{equation}
From the commutation formula for pseudodifferential operators (see Theorem~\ref{t:composition}), we know that 
$$\left[-\frac{\hbar^2\Delta}{2},\Op_\hbar\left(\ml{I}_{\Lambda}(a)\left(x,\xi,\frac{\hbar}{\eps_{\hbar}H_{\Lambda}(\xi)}\right)\right)\right]\hspace{4cm}$$
$$=\frac{\eps_{\hbar}}{i}\Op_\hbar\left(\left(\frac{\hbar H_{\Lambda}(\xi)}{\eps_{\hbar}}\frac{\vec{v}_{\Lambda}}{L_{\Lambda}}\cdot\partial_x 
+\frac{\hbar^2}{2i\eps_{\hbar}}\left(\frac{\vec{v}_{\Lambda}}{L_{\Lambda}}\cdot\partial_x\right)^2 \right)
\ml{I}_{\Lambda}(a)\left(x,\xi,\frac{\hbar}{\eps_{\hbar}H_{\Lambda}(\xi)}\right)\right),$$
where we used that $\ml{I}_{\Lambda}(a)$ has only Fourier coefficients along $\Lambda$. Comining this formula with~\eqref{deriv3}, we find that
$$\frac{d}{dt}\la F_{\Lambda,\hbar}(t\tau_{\hbar}^c),\ml{I}_{\Lambda}(a)\ra=
\left\la F_{\Lambda,\hbar}(t\tau_{\hbar}^c),
\left(\eta\frac{\vec{v}_{\Lambda}}{L_{\Lambda}}\cdot\partial_x
+\frac{\hbar^2}{2i\eps_{\hbar}}\left(\frac{\vec{v}_{\Lambda}}{L_{\Lambda}}\cdot\partial_x\right)^2+i V\right)\ml{I}_{\Lambda}(a)\right\ra.$$
This can be rewritten as
$$\frac{1}{i}\frac{d}{dt}\la F_{\Lambda,\hbar}(t\tau_{\hbar}^c),\ml{I}_{\Lambda}(a)\ra=
\left\la F_{\Lambda,\hbar}(t\tau_{\hbar}^c),
\left(\eta\frac{\vec{v}_{\Lambda}}{iL_{\Lambda}}\cdot\partial_x
-\frac{\hbar^2}{2\eps_{\hbar}}\left(\frac{\vec{v}_{\Lambda}}{L_{\Lambda}}\cdot\partial_x\right)^2+ V\right)\ml{I}_{\Lambda}(a)\right\ra.$$
We can now integrate this relation against a smooth function $\theta$ in $\ml{C}^{\infty}_c(\IR)$. We use that $\eps_{\hbar}\gg\hbar^2$ and 
we find that, for every $\theta$ in $\ml{C}^{\infty}_c(\IR)$,
 $$-\int_{\IR}\theta'(t)\la \tilde{F}_{\Lambda}(t),\ml{I}_{\Lambda}(a)\ra dt=
\int_{\IR}\theta(t)\left\la \tilde{F}_{\Lambda}(t),
\left(\eta\frac{\vec{v}_{\Lambda}}{L_{\Lambda}}\cdot\partial_x+ iV\right)\ml{I}_{\Lambda}(a)\right\ra dt.$$
From Lemma~\ref{l:invariance-geodesic-flow-two-microlocal} and as $\tilde{F}_{\Lambda}(t)$ is a finite complex Radon measure 
supported in $\IT^2\times\Lambda^{\perp}-\{0\}\times\IR$, we find that this is equivalent to
 $$-\int_{\IR}\theta'(t)\la \tilde{F}_{\Lambda}(t),\ml{I}_{\Lambda}(a)\ra dt=
\int_{\IR}\theta(t)\left\la \tilde{F}_{\Lambda}(t),
\left(\eta\frac{\vec{v}_{\Lambda}}{L_{\Lambda}}\cdot\partial_x+ i\ml{I}_{\Lambda}(V)\right)\ml{I}_{\Lambda}(a)\right\ra dt,$$
which implies Proposition~\ref{p:propag-compact-part}.
\end{proof}

\subsection{Conclusion}\label{ss:conclusion} Recall that we are primarly interested in describing the quantum fidelity distribution at the critical time scale 
$\tau_{\hbar}^c=\frac{\hbar}{\eps_{\hbar}}$. From Proposition~\ref{p:decomp-fidelity} and~\eqref{e:support-2micro}, 
one has
\begin{eqnarray*}
\la F(t),1\ra &=&e^{it\int_{\IT^2} V}\la F(0),1\ra\\
&+ &ie^{it\int_{\IT^2} V}\int_0^te^{-is\int_{\IT^2} V}\sum_{\Lambda:\text{rk}(\Lambda)=1}\left\la\int_{\IR^2\times\IR}\ml{I}_{\Lambda}(\tilde{F}_{\Lambda})(s,d\xi,d\eta),
\ml{I}_{\Lambda}(V)-\int_{\IT^2} V\right\ra ds. 
\end{eqnarray*}
Thanks to Proposition~\ref{p:propag-compact-part}, this can be rewritten as
\begin{eqnarray*}
\la F(t),1\ra &=&e^{it\int_{\IT^2} V}\la F(0),1\ra\\
&+ &e^{it\int_{\IT^2} V}\int_0^t\sum_{\Lambda:\text{rk}(\Lambda)=1}\left\la\int_{\IR^2\times\IR}\ml{I}_{\Lambda}(\tilde{F}_{\Lambda}^0),\frac{d}{ds}\left(
e^{i\int_0^s\left(\ml{I}_{\Lambda}(V)\left(x+s'\frac{\eta \vec{v}_{\Lambda}}{L_{\Lambda}}\right)-\int_{\IT^2}Vdx\right)ds'}\right)\right\ra ds, 
\end{eqnarray*}
which yields
\begin{eqnarray*}
\la F(t),1\ra &=&e^{it\int_{\IT^2} V}\left(\la F(0),1\ra-\sum_{\Lambda:\text{rk}(\Lambda)=1}\la \tilde{F}_{\Lambda}^0,1\ra\right)\\
&+ &\sum_{\Lambda:\text{rk}(\Lambda)=1}\int_{T^*\IT^2\times\IR}
e^{i\int_0^t\ml{I}_{\Lambda}(V)\left(x+s\frac{\eta \vec{v}_{\Lambda}}{L_{\Lambda}}\right)ds}\ml{I}_{\Lambda}(\tilde{F}_{\Lambda}^0)(dx,d\xi,d\eta). 
\end{eqnarray*}
Suppose now that we take $\psi_{\hbar}^1=\psi_{\hbar}^2=\psi_{\hbar}$ for the sequence of initial data. We are now in the framework of the introduction and we can make use 
of~\eqref{e:pushforward-2micro} to write
\begin{eqnarray*}
\la F(t),1\ra &=&e^{it\int_{\IT^2} V}\left(\la F(0),1\ra-\sum_{\Lambda:\text{rk}(\Lambda)=1}\la \mathbf{F}_{\Lambda}^0,1\ra\right)\\
&+ &\sum_{\Lambda:\text{rk}(\Lambda)=1}\int_{T^*\IT^2\times\IR}
e^{i\int_0^t\ml{I}_{\Lambda}(V)\left(x+s\frac{\eta \vec{v}_{\Lambda}}{L_{\Lambda}}\right)ds}\mathbf{F}_{\Lambda}^0(dx,d\eta), 
\end{eqnarray*}
which is exactly the content of Theorem~\ref{t:evolution} as $F(0)=1$ in that case.

\appendix

\section{The case of strong perturbations: $\eps_{\hbar}\geq \hbar$}
\label{a:strong}

All along the article, we focused on the case where $\eps_{\hbar}$ satisfies~\eqref{e:perturbation-size}. In particular, we had that $\eps_{\hbar}\ll\hbar$ which ensured that certain invariance 
properties of the fidelity distribution are satisfied -- namely Lemma~\ref{l:invariance-geodesic-flow}. In this appendix, we will briefly explain what can be said when the strength of the perturbation verifies
$$\hbar\leq\eps_{\hbar}\leq 1.$$
In that case, one has $\tau_{\hbar}^c=\frac{\hbar}{\eps_{\hbar}}\leq 1$ and we are in a scale of times where semiclassical rules for pseudodifferential operators apply. Precisely, we prove:
\begin{theo}\label{t:strong} Denote $u_{\hbar}(t')$ (resp. $u_{\hbar}^{\eps}(t')$) the solution to~\eqref{e:schrodinger-unperturbed} 
(resp.~\eqref{e:schrodinger-perturbed}) with 
initial condition $\psi_{\hbar}$ generating an unique semiclassical measure $F_0$. Then, one has
\begin{enumerate}
 \item if $\eps_{\hbar}=\hbar$, then, for every $t\in\IR$,
$$\lim_{\hbar\rightarrow 0^+}\left|\la u_{\hbar}^{\eps}(t\tau_{\hbar}^c), u_{\hbar}(t\tau_{\hbar}^c)\ra_{L^2(\IT^2)}\right|^2=\left|\la F_0,e^{i\int_0^t V\circ\varphi^s ds}\ra\right|^2,$$
 \item if $\hbar\ll\eps_{\hbar}\leq 1$, then, for every $t\in\IR$,
 $$\lim_{\hbar\rightarrow 0^+}\left|\la u_{\hbar}^{\eps}(t\tau_{\hbar}^c), u_{\hbar}(t\tau_{\hbar}^c)\ra_{L^2(\IT^2)}\right|^2=\left|\la F_0,e^{it V}\ra\right|^2.$$
\end{enumerate}

\end{theo}

Here, we focus on the case of the $2$-dimensional torus as we did in the rest of the article. Yet, the proof we give can be adapted verbatim to any smooth compact Riemannian manifold $(M,g)$ while this 
was not the case for smaller perturbations.
\begin{proof} We can write
 $$ \la u_{\hbar}^{\eps}(t\tau_{\hbar}^c), u_{\hbar}(t\tau_{\hbar}^c)\ra=
 \left\la\psi_{\hbar}, e^{\frac{it\tau_{\hbar}^c\widehat{P}_{\eps}(\hbar)}{\hbar}}e^{-\frac{it\tau_{\hbar}^c\widehat{P}_{0}(\hbar)}{\hbar}}\psi_{\hbar}\right\ra.$$ 
Compared with the rest of the article, we have here $\tau_{\hbar}^c=\frac{\hbar}{\eps_{\hbar}}\leq 1$. In particular, we can apply the Egorov Theorem to determine the value of $F(t)$ where the only difference 
with the classical case is that we have $e^{\frac{it\tau_{\hbar}^c\widehat{P}_{\eps}(\hbar)}{\hbar}}$ on one side and $e^{-\frac{it\tau_{\hbar}^c\widehat{P}_{0}(\hbar)}{\hbar}}$ on the other. Yet, the classical 
proof can be adapted to encompass this case\footnote{The proof in that reference is given for $V$ of the form $iW$ with $W$ real valued but it can be adapted to treat 
the selfadjoint case which is actually simpler to deal with.}~\cite[p.~71-72]{Ro10}. More precisely, we have
$$e^{\frac{it\tau_{\hbar}^c\widehat{P}_{\eps}(\hbar)}{\hbar}}e^{-\frac{it\tau_{\hbar}^c\widehat{P}_{0}(\hbar)}{\hbar}}=\Oph\left(
e^{i\frac{\eps_{\hbar}}{\hbar}\int_0^{t\tau_{\hbar}^c} V\circ\varphi^{s}ds}\right)+o_{L^2\rightarrow L^2}(1).$$
This yields
\begin{itemize}
 \item if $\eps_{\hbar}=\hbar$, then one has
 $$\la F(t),1\ra=\left\la F_0, e^{i\int_0^{t} V\circ\varphi^{s}ds}\right\ra.$$
 \item if $\hbar\ll\eps_{\hbar}\leq 1$, then one has
 $$\la F(t),1\ra=\left\la F_0, e^{itV}\right\ra.$$
\end{itemize}
Finally, we find that, for $\eps_{\hbar}=\hbar$ and for every $t$ in $\IR$,
$$\lim_{\hbar\rightarrow 0^+}\left|\la u_{\hbar}^{\eps}(t\tau_{\hbar}^c), u_{\hbar}(t\tau_{\hbar}^c)\ra_{L^2(\IT^2)}\right|^2=\left|\la F_0,e^{i\int_0^t V\circ\varphi^s ds}\ra\right|^2,$$
while, for $\hbar\ll\eps_{\hbar}\leq 1$ and for every $t$ in $\IR$,
$$\lim_{\hbar\rightarrow 0^+}\left|\la u_{\hbar}^{\eps}(t\tau_{\hbar}^c), u_{\hbar}(t\tau_{\hbar}^c)\ra_{L^2(\IT^2)}\right|^2=\left|\la F_0,e^{it V}\ra\right|^2.$$

\end{proof}

\section{Background on semiclassical analysis}
\label{a:sc-an}

In this appendix, we give a brief reminder on semiclassical analysis and we refer to~\cite{Zw12} (mainly Chapters $1$ to $5$) for a more detailed exposition. Given $\hbar>0$ 
and $a$ in $\ml{S}(\IR^{2d})$ (the Schwartz class), one can define the standard quantization of $a$ as follows:
$$\forall u\in\ml{S}(\IR^d),\ \Oph(a)u(x):=\frac{1}{(2\pi\hbar)^d}\iint_{\IR^{2d}}e^{\frac{i}{\hbar}\la x-y,\xi\ra}a\left(x,\xi\right)u(y)dyd\xi.$$
\begin{rema}
We could use other quantization procedures like the Weyl's one~\cite{Zw12}. The advantage of this quantization is that it has a simple action on trigonometric polynomials 
and that it behaves nicely with respect to multiplication by $V(x)$.
\end{rema}
This definition can be extended to any observable $a$ with uniformly bounded derivatives, i.e. such that for every $\alpha\in\IN^{2d}$, there exists $C_{\alpha}>0$ such that 
$\sup_{x,\xi}|\partial^{\alpha}a(x,\xi)|\leq C_{\alpha}$. More generally, we will use the convention, for every $m\in\IR$ and every $k\in\IZ$,
$$S^{m,k}:=\left\{(a_{\hbar}(x,\xi))_{0<\hbar\leq 1}:\ \forall (\alpha,\beta)\in\IN^d\times\IN^d,\ \sup_{(x,\xi)\in\IR^{2d}; 0<\hbar\leq 1}|\hbar^k\la\xi\ra^{-m}\partial_x^{\alpha}\partial_{\xi}^{\beta}a_{\hbar}(x,\xi)|<+\infty\right\},$$
where $\la\xi\ra:=(1+\|\xi\|^2)^{1/2}$. For such symbols, $\Oph(a)$ defines a continuous operator $\ml{S}(\IR^d)\rightarrow\ml{S}(\IR^d)$. 

\begin{rema} We also note that we have the following relation that we use at different stages of our proof:
\begin{equation}\label{e:change-variable}
\forall\delta>0,\ \forall a\in S^{m,k}, \Oph(a(x,\xi)):=\Op_{\hbar\delta^{-1}}(a(x,\delta\xi)).
\end{equation}
\end{rema}

Among the above symbols, we distinguish the family of $\IZ^d$-periodic symbols that we denote by $S^{m,k}_{per}$. Note that any $a$ in $\ml{C}^{\infty}(T^*\IT^d)$ 
(with bounded derivatives) defines an element in $S^{0,0}_{per}$. According to Th.~$4.19$ in~\cite{Zw12}, for any $a\in S^{m,k}_{per}$, the operator $\Oph(a)$ maps trigonometric polynomials into a smooth $\IZ^d$-periodic function, and more generally any smooth $\IZ^d$-periodic function into a smooth $\IZ^d$-periodic function. Thus, for every $a$ in $S^{m,k}_{per}$, $\Oph(a)$ acts by duality on the space of distributions $\ml{D}'(\IT^d)$. An important feature of this quantization procedure is that it defines a bounded operator on $L^2(\IT^d)$:
\begin{theo}[Calder\'on-Vaillancourt]\label{t:cald-vail} There exists a constant $C_d>0$ and an integer $D>0$ such that, for every $a$ in $S^{0,0}_{per}$, one has, for every $0<\hbar\leq 1$,
$$\left\|\Oph(a)\right\|_{L^2(\IT^d)\rightarrow L^2(\IT^d)}\leq C_d\sum_{|\alpha|\leq d+1}\|\partial^{\alpha}_xa\|_{\infty},$$
and
$$\left\|\Oph(a)\right\|_{L^2(\IT^d)\rightarrow L^2(\IT^d)}\leq C_d\sum_{|\alpha|\leq D}\hbar^{\frac{|\alpha|}{2}}\|\partial^{\alpha}a\|_{\infty}.$$
\end{theo}
The second part of the Theorem is in fact the ``standard'' Calder\'on-Vaillancourt Theorem whose proof can be found\footnote{The proof in this reference is given for the Weyl quantization 
but the argument can be adapted for the standard quantization.} 
in~\cite[Th.~5.5]{Zw12} while the first part can be found in~\cite{Ru10}.
\begin{proof}
We recall the proof of the first part of the Theorem which is maybe less known~\cite{Ru10} and we refer to~\cite[Chap.~5]{Zw12} for the second part of the Theorem. 
The first part follows from the fact that, for every $k$ in $\IZ^d$, one has
$$\left(\Op_{\hbar}\left(a\right)e_k\right)(x)=a(x,2\pi\hbar k)e_k(x),$$
where $e_k(x):=e^{2i\pi k.x}.$ The proof of this fact is given in chapter~$4$ of~\cite[Th.~$4.19$]{Zw12}. Once we have observed this, we can write, for every trigonometric polynomial $u$ in 
$L^2(\IT^d)$, its Fourier decomposition $u=\sum_{k\in\IZ^d}\widehat{u}_ke_k$, and
$$\Op_{\hbar}\left(a\right)u=\sum_{k,l\in\IZ^d}\widehat{u}_k\widehat{a}(l,2\pi\hbar k)e_{k+l}=\sum_{p\in\IZ^d}e_p\sum_{k\in\IZ^d}\widehat{u}_k\widehat{a}(p-k,2\pi\hbar k),$$
where $a(x,\xi)=\sum_{l\in\IZ^d}\widehat{a}(l,\xi)e_l(x).$ Applying Plancherel equality, we get
$$\left\|\Op_{\hbar}\left(a\right)u\right\|_{L^2(\IT^d)}^2=\sum_{p\in\IZ^d}\left|\sum_{k\in\IZ^d}\widehat{u}_k\widehat{a}(p-k,2\pi\hbar k)\right|^2.$$
Thanks to Cauchy-Schwarz inequality, one has 
$$\left\|\Op_{\hbar}\left(a\right)u\right\|_{L^2(\IT^d)}^2\leq\sum_{p\in\IZ^d}\left(\sum_{k\in\IZ^d}|\widehat{u}_k|^2|\widehat{a}(p-k,2\pi\hbar k)|\right)
\left(\sum_{k'\in\IZ^d}|\widehat{a}(p-k',2\pi\hbar k')|\right).$$
This implies that
$$\left\|\Op_{\hbar}\left(a\right)\right\|_{L^2(\IT^d)\rightarrow L^2(\IT^d)}^2\leq\sup_{p\in\IZ^d}\left(\sum_{k'\in\IZ^d}|\widehat{a}(p-k',2\pi \hbar k')|\right)\times
 \sup_{k\in\IZ^d}\left(\sum_{p\in\IZ^d}|\widehat{a}(p-k,2\pi\hbar k)|\right),$$
which concludes the proof of the lemma. 

\end{proof}

Another important feature of this quantization procedure is the composition formula:
\begin{theo}[Composition formula]\label{t:composition} Let $a\in S^{m_1,k_1}$ and $b\in S^{m_2,k_2}$. Then, one has, for any $0<\hbar\leq 1$
$$\Oph(a)\circ\Oph(b)=\Oph(a\sharp_{\hbar} b),$$
 in the sense of operators from $\ml{S}(\IR^d)\rightarrow\ml{S}(\IR^d)$, where $a\sharp_{\hbar} b$ has uniformly bounded derivatives, and, for every $N\geq 0$
$$a\sharp_{\hbar} b\sim \sum_{k=0}^N\frac{1}{k!}\left(\frac{\hbar}{i}D\right)^k(a,b)+\ml{O}(\hbar^{N+1}),$$
where $D(a,b)(x,\xi)=(\partial_{y}.\partial_{\xi})(a(x,\xi)b(y,\nu))\rceil_{x=y,\xi=\nu}$.
\end{theo}
We refer to chapter~$4$ of~\cite{Zw12} for a detailed proof of this result. We observe that for $N=0$, the coefficient is given by the symbol $ab$. As before, we can restrict this result to the case of periodic symbols, and we can check that the composition formula remains valid for operators acting on $\ml{C}^{\infty}(\IT^d)$.
\begin{rema}\label{r:symmetry} We note that we have in fact the following useful property. If $a(\xi)$ is a polynomial in $\xi$ of order $\leq N$, one has, the exact formula:
$$a\sharp_{\hbar} b-b\sharp_{\hbar} a=a\sharp_{\hbar} b=\sum_{k=0}^{N}\frac{1}{k!}\left(\frac{\hbar}{i}D\right)^k(a,b).$$

\end{rema}


\begin{thebibliography}{99}
\bibitem{AFM12} N.~Anantharaman, C.~Fermanian-Kammerer, F.~Maci\`a \emph{Semiclassical completely integrable systems: long-time dynamics and observability via two-microlocal Wigner measures},  Amer. J. Math. $\mathbf{137}$ 
(2015), 577--638
\bibitem{AL14} N.~Anantharaman, M.~L\'eautaud, \emph{Sharp polynomial decay rates for the damped wave equation on the torus}, Anal. PDE $\mathbf{7}$ (2014), 159--214
\bibitem{ALM16} N.~Anantharaman, M.~L\'eautaud, F.~Maci\`a, \emph{Wigner measures and observability for the Schr\"odinger equation on the disk}, Invent. Math. $\mathbf{206}$ (2016), 485--599
\bibitem{AM10} N.~Anantharaman, F.~Maci\`a \emph{Semiclassical measures for the Schr\"odinger equation on the torus},  J. Eur. Math. Soc. (JEMS) $\mathbf{16}$ (2014), 1253--1288
\bibitem{BolSc06} J.~Bolte, T.~Schwaibold \emph{Stability of wave packet dynamics under perturbations}, 
Phys. Rev. E73 (2006) 026223. 
\bibitem{CanJaTo12} Y. Canzani, D.~Jakobson, J.~Toth \emph{On the distribution of perturbations of propagated Schr\"odinger eigenfunctions},  J. of Spectral Theory $\mathbf{4}$ (2014), 283--307
\bibitem{CoRo07} M.~Combescure, D.~Robert \emph{A phase-space study of the quantum Loschmidt Echo in the semiclassical limit}, 
Ann. H. Poincar\'e $\mathbf{8}$ (2007), 91--108
\bibitem{DuGo14} R.~Dubertrand, A.~Goussev \emph{Origin of the exponential decay of the Loschmidt echo in integrable systems}, Phys. Rev. E $\mathbf{89}$ (2014), 022915 
\bibitem{EsRi15} S.~Eswarathasan, G.~Rivi\`ere \emph{Perturbation of the semiclassical Schr\"odinger equation on negatively curved surfaces}, 
J. Inst. Math. Jussieu $\mathbf{16}$ (2017), 787--835
\bibitem{EsTo12} S.~Eswarathasan, J.~Toth \emph{Average pointwise bounds for deformations of Schrodinger eigenfunctions }, Ann. H. Poincar\'e $\mathbf{14}$ (2012), 611--637
\bibitem{Fe00} C. Fermanian-Kammerer \emph{Mesures semi--classiques 2-microlocales}, C. R. Acad. Sci. Paris Ser. I
Math., $\mathbf{331}$ (2000), 515--518
\bibitem{FeGe02} C. Fermanian-Kammerer, P. G\'erard \emph{Mesures semi-classiques et croisement de modes}, Bull. Soc. Math. France $\mathbf{130}$ (2002), 123--168
\bibitem{Ge91} P.~G\'erard \emph{Mesures semi-classiques et ondes de Bloch}, Sem. EDP (Polytechnique) 1990--1991, Exp. 16 (1991)
\bibitem{GPSZ06} T.~Gorin, T.~Prosen, T.H.~Seligman, M.~Zdinaric \emph{Dynamics of Loschmidt echoes and fidelity decay}, Physics Reports $\mathbf{435}$ (2006) 33--156
\bibitem{GJPW12} A.~Goussev, R.A.~Jalabert, H.M.~Pastawski, D.~Wisniacki \emph{Loschmidt Echo}, Scholarpedia 7(8), 11687,  arXiv:1206.6348 (2012)
\bibitem{JaPe09} P.~Jacquod, C.~Petitjean \emph{Decoherence, Entanglement and Irreversibility in Quantum Dynamical
Systems with Few Degrees of Freedom}, Adv. Phys. $\mathbf{58}$, 67--196 (2009)
\bibitem{Ma09} F.~Maci\`a \emph{Semiclassical measures and the Schr\"odinger flow on Riemannian manifolds}, Nonlinearity $\mathbf{22}$ (2009), 1003--1020
\bibitem{Ma10} F.~Maci\`a \emph{High-frequency propagation for the Schr\"odinger equation on the torus}, Jour. Funct. Analysis $\mathbf{258}$ (2010), 933--955
\bibitem{MaRi16} F.~Maci\`a, G.~Rivi\`ere \emph{Concentration and non concentration for the Schr\"odinger evolution on Zoll manifolds}, Comm. In Math. Phys.
$\mathbf{345}$ (2016), 1019--1054
\bibitem{MaRi16b} F.~Maci\`a, G.~Rivi\`ere, \emph{Two microlocal regularity of quasimodes on the torus}, preprint arXiv:1708.08799 (2017)
\bibitem{Mi96} L. Miller \emph{Propagation d'ondes semi--classiques \`a travers une interface et mesures $2$--microlocales}, PhD thesis, Ecole polytechnique, 
Palaiseau (1996)
\bibitem{Ni96} F.~Nier \emph{A semiclassical picture of quantum scattering}, Ann. Sci. ENS $\mathbf{29}$ (1996), 149--183
\bibitem{Pe84} A.~Peres \emph{Stability of quantum motion in chaotic and regular systems}, Phys. Rev. A $\mathbf{30}$ (1984), 1610--1615
\bibitem{Ri16} G.~Rivi\`ere \emph{Long time dynamics of the perturbed Schr\"odinger equation on negatively curved surfaces}, Ann. H. Poincar\'e $\mathbf{17}$ (2016), 1955--1999
\bibitem{Ro10} J.~Royer \emph{Analyse haute fr\'equence de l'\'equation de Helmholtz dissipative}, Th\`ese Universit\'e de Nantes tel--00578423 (2010)
\bibitem{Ru10} M.~Ruzhansky, V.~Turunen \emph{Pseudodifferential operators and symmetries}, Birkh\"auser Verlag, Basel Boston Berlin (2010)
\bibitem{Zw12} M.~Zworski \emph{Semiclassical analysis}, Graduate Studies in Mathematics $\mathbf{138}$, AMS (2012)
\end{thebibliography}
\end{document}